\documentclass[a4paper]{article}
\usepackage{amssymb,amsmath,amsthm}
\usepackage{multirow}
\usepackage{latexsym}
\usepackage{mathrsfs}
\usepackage{stmaryrd}
\usepackage{color}
\usepackage[urlbordercolor={.8 0 0}]{hyperref}
\usepackage{geometry}

\usepackage{proof}
\usepackage{graphicx}

\newtheorem{Prop}{Proposition}[section]
\newtheorem{Lemm}[Prop]{Lemma}
\newtheorem{Th}[Prop]{Theorem}
\newtheorem{Cor}[Prop]{Corollary}
\newtheorem{Folk}[Prop]{Folklore}

\theoremstyle{definition}
\newtheorem{Def}[Prop]{Definition}

\newcommand{\BS}{\mathop{\multimap}}
\newcommand{\SL}{\mathop{\mbox{\raisebox{.5em}{\rotatebox{180}{$\multimap$}}}}}

\newcommand{\Bs}{\mathop{\backslash}}
\newcommand{\Sl}{\mathop{/}}

\newcommand{\ncplus}{\boxplus}

\newcommand{\yields}{\vdash}
\newcommand{\aconj}{\mathop{\&}}
\newcommand{\adisj}{\oplus}
\newcommand{\mconj}{\otimes}

\newcommand{\KStar}{{}^*}

\newcommand{\U}{\mathbf{1}}
\newcommand{\Z}{\mathbf{0}}

\newcommand{\ID}{\mathrm{id}}
\newcommand{\PERM}{\mathrm{perm}}
\newcommand{\WEAK}{\mathrm{weak}}
\newcommand{\NCONTR}{\mathrm{ncontr}}
\newcommand{\CUT}{\mathrm{cut}}

\newcommand{\NN}{\mathbb{N}}

\newcommand{\Ic}{\mathcal{I}}
\newcommand{\Wc}{\mathcal{W}}
\newcommand{\Cc}{\mathcal{C}}
\newcommand{\Ec}{\mathcal{E}}

\newcommand{\Nc}{\mathcal{N}}

\newcommand{\Dcut}{\mathscr{D}}
\newcommand{\Dnocut}{\hat{\mathscr{D}}}
\newcommand{\Scut}{\mathcal{S}}
\newcommand{\Snocut}{\hat{\mathcal{S}}}

\newcommand{\ACTomega}{\mathrm{ACT}_\omega}
\newcommand{\eACTomega}{{!}\ACTomega}

\begin{document}

\title{Infinitary Action Logic with Exponentiation}

\author{Stepan L.\ Kuznetsov \and Stanislav O.~Speranski}


\maketitle


\begin{abstract}
 We introduce infinitary action logic with exponentiation---that is, the multiplicative-additive Lambek calculus extended with Kleene star and with a family of subexponential modalities, which allows some of the structural rules (contraction, weakening, permutation). The logic is presented in the form of an infinitary sequent calculus. We prove cut elimination and, in the case where at least one subexponential allows non-local contraction, establish exact complexity boundaries in two senses. First, we show that the derivability problem for this logic is $\Pi_1^1$-complete. Second, we show that the closure ordinal of its derivability operator is $\omega_1^{\mathrm{CK}}$. In the case where no subexponential allows contraction, we show that complexity is the same as for infinitary action logic itself. Namely, the derivability problem in this case is $\Pi^0_1$-complete and the closure ordinal is not greater than $\omega^\omega$.
\end{abstract}

%
%


\section{Linguistic introduction}

The {\em Lambek calculus} was introduced in \cite{Lambek-1958} as a logical background for describing~na\-tu\-ral language syntax.
Lambek's approach was built upon earlier ideas of {\em categorial grammar} from~\cite{Ajdukiewicz-1935} and \cite{Bar-Hillel-1953}.

The two principal operations
of the Lambek calculus are two {\em divisions,} left and right. Left division, $A \Bs B$ (``$B$ divided by $A$ to the left,'' or
``$A$ under $B$'' for short), denotes the set of words which, being prefixed with any word from $A$, form words of $B$. For example,
if $NP$ is the language of all English noun phrases (like ``John'' or ``the red ball'') and $S$ includes all grammatically correct English sentences,
then $NP \Bs S$ includes the language of intransitive verbs. Indeed, if $v$ is an intransitive verb, then for any noun phrase $n$ the
concatenation $nv$ is a valid sentence, like ``John runs,'' for example.

The right division, $B \Sl A$ (``$B$ divided by $A$ to the right,'' or ``$B$ over $A$'' for short), is defined symmetrically. Thus, for example,
$(NP \Bs S) \Sl NP$ includes the language of transitive verbs. A transitive verb requires two noun phrases, one from each side, to become a complete sentence
({\em e.g.,} ``John loves Mary''). Formulae like $(NP \Bs S) \Sl NP$, which define languages in terms of basic ones (like $NP$ and $S$), are called {\em syntactic types.}

Going further (see, for example, Carpenter's textbook~\cite{Carpenter}), one defines syntactic types for other parts of speech:
\begin{center}
\begin{tabular}{l@{\quad$\rhd$\quad}l}
 common noun (``book,'' ``paper,'' ``girl,'' ...) & $N$ \\[3pt]
 noun phrase (``John,'' ``Mary,'' ``the book,'' ``a girl,'' ...) & $NP$ \\[3pt]
 article (``the,'' ``a'') & $NP \Sl N$ \\[3pt]
 transitive verb (``loves,'' ``signed,'' ...) & $(NP \Bs S) \Sl NP$ \\[3pt]
 intransitive verb (``runs,'' ``sleeps,'' ...) & $NP \Bs S$ \\[3pt]
 adjective (``red,'' ``interesting,'' ...) & $N \Sl N$ \\[3pt]
 adverb (``very,'' ``extremely,'' ...) & $(N \Sl N) \Sl (N \Sl N)$
\end{tabular}
\end{center}
The new basic syntactic type $N$ for common nouns is needed since English distinguishes them from noun phrases. Thus, only a noun phrase can be used as a subject, but, on the other hand, adjectives act as modifiers for common nouns, not noun phrases. A common noun gets transformed into a noun phrase by an article.

Let $A \preccurlyeq B$ mean that $B$ is a broader syntactic type than $A$.  Then Lambek's division operations obey the following conditions (product, $\cdot$,
means concatenation):
$$
A \preccurlyeq C \Sl B \iff A \cdot B \preccurlyeq C \iff B \preccurlyeq A \Bs C
$$
which, along with the associativity of product, and the reflexivity and transitivity of $\preccurlyeq$, form the Lambek calculus in its
non-sequential form.

The interpretation of Lambek divisions sketched above can be made formal by considering so-called \emph{language models} (or \emph{L-models})
for the Lambek calculus.
We fix an alphabet $\Sigma$ and interpret Lambek formulae as formal languages over this alphabet. Multiplication corresponds to pairwise
concatenation, and division operations are intepreted as follows:
\begin{align*}
{A \Bs B}\ &=\ {\{ u \mid (\forall v \in A) \, vu \in B\}},\\
{B \Sl A}\ &=\ {\{ u \mid (\forall v \in A) \, uv \in B\}} .
\end{align*}

Notice that the original Lambek calculus features the so-called \emph{Lambek non-emptiness restriction}. In L-models it corresponds to excluding
the empty word from all the languages considered. In particular, this is important in definitions of division operations on languages.
If one allows the empty word, it will be always included into $A \Sl A$, which yields $B \Sl (A \Sl A) \preccurlyeq B$.  An unwanted linguistic
consequence of this principle is $(N \Sl N) \Sl (N \Sl N) \preccurlyeq N \Sl N$, which yields $((N \Sl N) \Sl (N \Sl N)) \cdot N \preccurlyeq N$
and validates ``extremely book'' as a valid noun phrase (cf. ``extremely interesting book,'' which is correct and analyzed as
$((N \Sl N) \Sl (N \Sl N)) \cdot (N \Sl N) \cdot N \preccurlyeq N$). In other words, the empty word acts as an adjective (like ``interesting'') and
accepts an adverb ``extremely.'' This is unacceptable by English grammar. (This example was taken from~\cite{MootRetore-2012}.)

In the view of the above, Lambek's non-emptiness restriction is a desirable feature of a logical system underlying categorial grammars.
Unfortunately, as shown in~\cite{KKS-2016-arXiv}, this restriction conflicts with one of the extensions of the Lambek calculus we study
in this paper, namely, (sub)exponential modalities (see below)---so we opt for the system without Lambek's restriction. This system is
called \emph{the Lambek calculus allowing empty antecedents}~\cite{Lambek-1961} and is a fragment of the broader system ${!}\mathrm{ACT}_\omega$,
which is defined in the next section.


As shown in~\cite{Pentus-1993}, Lambek grammars can describe only context-free languages.
From the linguistical point of view, it is a serious limitation, since there exist natural language
phenomena that go beyond the context-free framework
(a formally justified example, based on Swiss German material, was provided in~\cite{Shieber-1985}).
The pursuit of expanding expressive capabilities of Lambek grammars motivates the study of various extensions and variations of the original Lambek system.

The first natural extension of the Lambek calculus is the so-called ``full,'' or multiplicative-additive Lambek calculus obtained by adding
{\em additive} conjunction and disjunction, which correspond to intersection and union. This increases the expressive power of Lambek grammars:
with additives, they can describe finite intersections of context-free languages~\cite{Kanazawa-1992} and even a broader class
of languages generated by conjunctive context-free grammars~\cite{Kuznetsov-2013,KuznetsovOkhotin-2017}. No non-trivial upper bounds are known for the
class of languages generated by Lambek grammars with additives.

Notice that in the presence of both additives completeness w.r.t. L-models fails, due to the distributivity law for additives.
On one hand, the distributivity law is true
 under set-theoretic interpretation of additive conjunction and disjunction. On the other hand, as noticed by Ono and Komori~\cite{OnoKomori-1985}, it is not derivable without using contraction, in particular, in the multiplicative-additive
Lambek calculus. 

From the modern point of view~\cite{Abrusci-1990}, the Lambek calculus can be viewed as an intuitionistic and non-commutative variant of
Girard's linear logic~\cite{Girard-1987}. From this point of view, Lambek divisions $\Bs$ and $\Sl$ become directed implications;
throughout this paper we denote them by $\BS$ and $\SL$. Multiplication, $\cdot$, corresponds to multiplicative conjunction
(``tensor,'' $\mconj$); intersection and union are additive conjunction and disjunction respectively. Multiplicative disjunction
(``par,'' $\raisebox{.7em}{\rotatebox{180}{\&}}$) is not included into the Lambek calculus, because it does not have natural
linguistic meaning. Being a substructural system, the Lambek calculus lacks all structural rules (weakening, permutation, and contraction),
except implicit associativity. Structural rules, however, can be restored in a controlled way using the exponential modality (also borrowed
from linear logic) and its weaker variants, called \emph{subexponentials} (see below).


In our paper we focus on the extension of the Lambek calculus by means of iteration (Kleene star) and subexponential modalities.
Elements of linguistic analysis here are mostly based on the categorial grammar framework from~\cite{MorVal-2015}
  developed for the CatLog
parser.

Kleene star is one of the standard operations on formal languages, thus it looks very natural to introduce it into L-models and, thus,
into the Lambek calculus.  A linguistic usage of Kleene star is shown in~\cite{MorVal-2015} (where it is denoted by ${?}$ and called
``existential exponential''): in ``John, Bill, Mary, and Suzy'' the coordinator ``and'' receives syntactic type $(NP^* \Bs NP) \Sl NP$.
The Kleene star, as shown in the next section, is axiomatized by means of an $\omega$-rule, which raises algorithmic complexity of the
system to very high levels.
Morrill and Valent\'{\i}n, however, in order to avoid undecidability, formulate an incomplete set of rules for Kleene star.

Historically, Kleene star first appeared in the study of events, or actions within a transition system: in the original paper~\cite{Kleene}
it was used when describing events in neural networks.
If $A$ denotes a class of actions, then $A^*$ means actions of class $A$ repeated several (possibly zero) times;
\cite{Pratt-1991} proposes {\em action algebras,} an extension of
Kleene algebras with residuals. Though Pratt's work was independent from Lambek, these residuals actually coincide with Lambek divisions.
In the presence of residuals, usage of infinitary systems for axiomatizing Kleene star becomes inevitable, due to complexity reasons
(see below).\footnote{It
should be remarked that L-completeness of the Lambek calculus with Kleene star (without additives) is still an open problem.}

The second family of connectives we use to extend the Lambek calculus is the family of subexponential modalities, or subexponentials for short.
Their linguistic motivation, going back to Morrill, is as follows. The Lambek calculus itself has a limited capability of treating relativization,
or dependent clauses. For example, ``that'' in ``book that John read'' gets syntactic type $(N \Bs N) \Sl (S \Sl NP)$, because the dependent clause ``John read'' lacks a noun phrase (``John read the book'') to become a complete sentence ($S$).
The place where the lacking noun phrase should be placed is called a {\em gap:} ``John read [].''
In more complicated situations, however, this does not work: in the phrase ``book that John read yesterday'' the dependent clause
``John read yesterday'' has a gap in the middle: ``John read [] yesterday,'' and is neither of type $S \Sl NP$, nor $NP \Bs S$.
This syntactic phenomenon is called {\em medial extraction} and can be handled
 by adding a special modality, denoted by ${!}$, which allows permutation rules. Now ``that'' receives syntactic type
$(N \Bs N) \Sl (S \Sl {!}NP)$, and ``John read [] yesterday'' is indeed of type $S \Sl {!}NP$, since by permutation ${!}NP$ reaches its place to fill the gap.

There is a more sophisticated phenomenon called {\em parasitic extraction:} in the example ``paper that John signed without reading''
the dependent clause includes two gaps: ``John signed [] without reading []'' which should both be filled with the same instance of $NP$ (``the paper'').
This is handled by the contraction rule which allows duplicating ${!}NP$.


Compared to the exponential connective in Girard's linear logic, the ${!}$ connective discussed above allows permutation and contraction,
but not weakening, since weakening would give linguistically invalid examples like ``book that John loves Mary'' (``John loves Mary'' has type $S$ and by
weakening would have recieved type $S \Sl {!}NP$). Such modalities are called {\em subexponential.}
Moreover, we consider polymodal systems with a family of subexponentials (even when two subexponentials obey the same rules, they are not
necessarily equivalent). Such extensions of commutative linear logic were considered in~\cite{NigamMiller-2009}, and for the
non-commutative one---in~\cite{KKNS-mscs-2018}. Besides linguistic usage sketched above, linear logic systems with subexponentials (both
commutative and non-commutative) have applications in logical frameworks for formal computation system specification~\cite{Polakow-2000,PfenningSimmons-2009,KKNS-IJCAR-2018}.






\section{Infinitary action logic with subexponentials} \label{sec-calculus}

In this section we define ${!}\mathrm{ACT}_\omega$, that is, infinitary action logic extended with a family of subexponentials.
This is the main system we are going to study. Throughout this paper, we use linear logic notation for formulae of
${!}\mathrm{ACT}_\omega$, in order to avoid notation clashes with classical logic, which is going to be used
as meta-logic inside our proofs.

\smallskip
We start by fixing a \emph{subexponential signature} of the form
$$
\Sigma\ =\
{\langle \mathcal{I}, \preccurlyeq, \mathcal{W}, \mathcal{C}, \mathcal{E} \rangle}
$$
where:
\begin{itemize}

\item $\mathcal{I}$ is a finite set, whose elements are called \emph{subexponential labels};

\item $\preccurlyeq$ is a preorder on $\mathcal{I}$;

\item $\mathcal{W}$, $\mathcal{C}$ and $\mathcal{E}$ are subsets of $\mathcal{I}$, each of which is closed upward w.r.t.\
$\preccurlyeq$.\footnote{Thus if $s_1 \in \mathcal{W}$ and $s_1 \preccurlyeq s_2$, then $s_2 \in \mathcal{W}$; similarly for $\mathcal{C}$ and $\mathcal{E}$.}

\end{itemize}
Intuitively, subexponentials indexed by elements from $\mathcal{W}$, $\mathcal{C}$ and $\mathcal{E}$ allow weakening, contraction and exchange
(permutation) respectively. Since contraction, in its non-local form (see below), and weakening derive exchange, we explicitly postulate $\mathcal{W} \cap \mathcal{C} \subseteq
\mathcal{E}$.

Formulae are built from \emph{propositional variables} $p_1, p_2, p_3, \ldots$ and constants $\U$ and $\Z$ (\emph{multiplicative unit} and \emph{zero} respectively) using the following
connectives:
\begin{itemize}

\item \emph{left} and \emph{right implications}, $\BS$ and $\SL$, also called \emph{right} and \emph{left divisions};

\item \emph{product}, $\mconj$, also called \emph{multiplicative conjunction};

\item \emph{additive conjunction}, $\aconj$, and \emph{disjunction}, $\adisj$;

\item \emph{iteration}, $\KStar$, also called \emph{Kleene star};

\item \emph{subexponentials}, denoted by ${!}^s$ for each $s \in \mathcal{I}$.

\end{itemize}
Here iteration and subexponentials are unary, while the other connectives are binary.

Sequents are expressions of the form $\Pi \yields A$, where $A$ is a formula and $\Pi$ is a sequence of formulae (possibly empty). In what follows, by $A^n$ we shall denote the sequence $A, \ldots, A$ ($n$~times); $A^0$ is the empty sequence.

\medskip
The axioms and rules of our calculus are as follows.

$$
\infer[(\ID)]{A \yields A}{}
$$

$$
\infer[(\BS\yields)]{\Gamma, \Pi, A \BS B, \Delta \yields C}
{\Pi \yields A & \Gamma, B, \Delta \yields C}
\qquad
\infer[(\yields\BS)]{\Pi  \yields A \BS B}
{A, \Pi \yields B}
$$

$$
\infer[(\SL\yields)]{\Gamma, B \SL A, \Pi, \Delta \yields C}
{\Pi \yields A & \Gamma, B, \Delta \yields C}
\qquad
\infer[(\yields\SL)]{\Pi \yields B \SL A}
{\Pi, A \yields B}
$$

$$
\infer[(\mconj\yields)]{\Gamma, A \mconj B, \Delta \yields C}
{\Gamma, A, B, \Delta \yields C}
\qquad
\infer[(\yields\mconj)]{\Gamma, \Delta \yields A \mconj B}
{\Gamma \yields A & \Delta \yields B}
$$

$$
\infer[(\U\yields)]{\Gamma, \U, \Delta \yields C}{\Gamma, \Delta \yields C}
\qquad
\infer[(\yields\U)]{\yields\U}{}
\qquad
\infer[(\Z\yields)]{\Gamma,\Z,\Delta \yields C}{}
$$

$$
\infer[(\adisj\yields)]{\Gamma, A_1 \adisj A_2, \Delta \yields C}
{\Gamma, A_1, \Delta \yields C & \Gamma, A_2, \Delta \yields C}
\qquad
\infer[(\yields\adisj)_i\mbox{, $i = 1,2$}]
{\Pi \yields A_1 \adisj A_2}{\Pi \yields A_i}
$$

$$
\infer[(\aconj\yields)_i\mbox{, $i = 1,2$}]
{\Gamma, A_1 \aconj A_2, \Delta \yields C}
{\Gamma, A_i, \Delta \yields C}
\qquad
\infer[(\yields\aconj)]{\Pi \yields A_1 \aconj A_2}
{\Pi \yields A_1 & \Pi \yields A_2}
$$

$$
\infer[(\KStar\yields)_\omega]
{\Gamma, A^*, \Delta \yields C}
{(\Gamma, A^n, \Delta \yields C)_{n \in \NN}}
$$

$$
\infer[(\yields\KStar)_0]
{\yields A^*}{}
\qquad
\infer[(\yields\KStar)_n\mbox{, $n >0$ and each $\Pi_i$ is non-empty}]
{\Pi_1, \dots, \Pi_n \yields A^*}
{\Pi_1 \yields A & \dots & \Pi_n \yields A}
$$

$$
\infer[({!}\yields)]{\Gamma, {!}^s A, \Delta \yields C}
{\Gamma, A, \Delta \yields C}
\qquad
\infer[(\yields{!})\mbox{, every $s_i \succcurlyeq s$}]
{{!}^{s_1} A_1, \dots, {!}^{s_n} A_n \yields {!}^s B}
{{!}^{s_1} A_1, \dots, {!}^{s_n} A_n \yields B}
$$

$$
\infer[(\WEAK)\mbox{, $w \in \Wc$}]
{\Gamma, {!}^w A, \Delta \yields C}
{\Gamma, \Delta \yields C}
$$

$$
\infer[(\PERM)_1\mbox{, $e \in \Ec$}]
{\Gamma, {!}^e A, \Pi, \Delta \yields C}
{\Gamma, \Pi, {!}^e A, \Delta \yields C}
\qquad
\infer[(\PERM)_2\mbox{, $e \in \Ec$}]
{\Gamma, \Pi, {!}^e A, \Delta \yields C}
{\Gamma, {!}^e A, \Pi, \Delta \yields C}
$$

$$
\infer[(\NCONTR)_1\mbox{, $c \in \Cc$}]
{\Gamma, {!}^c A, \Pi, \Delta \yields C}
{\Gamma, {!}^c A, \Pi, {!}^c A, \Delta \yields C}
\qquad
\infer[(\NCONTR)_2\mbox{, $c \in \Cc$}]
{\Gamma, \Pi, {!}^c A, \Delta \yields C}
{\Gamma, {!}^c A, \Pi, {!}^c A, \Delta \yields C}
$$

$$
\infer[(\CUT)]
{\Gamma, \Pi, \Delta \yields C}
{\Pi \yields A & \Gamma, A, \Delta \yields C}
$$

\bigskip
\noindent
In the presence of an $\omega$-rule~--- namely $(\KStar\yields)_\omega$~--- derivability should be described with a certain amount of care.
There are different but equivalent ways of defining the notion of an (infinitary) derivation (cf.\ \cite[Definition~1.4.4]{Aczel-1977} and
\cite[\S~1]{Buchholz-1977}). For our purposes, we shall employ the following approach.

\begin{Def}
A {\em derivation} in $\eACTomega$ is a well-founded (though not necessarily finite branching) tree $\mathfrak{T}$ such that:
\begin{enumerate}

\item[a.] all the vertices of $\mathfrak{T}$ are labeled by sequents;

\item[b.] for each vertex of $\mathfrak{T}$, its label can be obtained from the labels of its children by an application~of
some rule in $\eACTomega$.\footnote{Here
axioms are treated as nullary rules; so the leafs of $\mathfrak{T}$ must be labeled by axioms.}

\end{enumerate}
In this situation, by the {\em goal} of $\mathfrak{T}$ we mean the label of the root of $\mathfrak{T}$. Naturally, a sequent $s$ is {\em derivable} in $\eACTomega$ iff
$s$ is the goal of some derivation in $\eACTomega$.
\end{Def}

An equivalent characterization of derivability in $\eACTomega$ is given by:

\begin{Prop}
The set of all derivable sequents in $\eACTomega$ coincides with the least set (with respect to inclusion) closed under the rules of $\eACTomega$.
\end{Prop}

\begin{proof}
Let $\mathcal{S}$ be the set of all sequents derivable in $\eACTomega$. Denote by $\mathscr{U}$ the collection of all sets of sequents that are
closed under the rules of $\eACTomega$. We wish to show that $\mathcal{S}$ in the least element of $\mathscr{U}$. Clearly, $\mathcal{S}$ belongs
to $\mathscr{U}$. Now consider an arbitrary $\mathcal{P} \in \mathscr{U}$. Given a derivation $\mathfrak{T}$ in $\eACTomega$, take
\[
V_{\mathfrak{T}}\ :=\
\text{the set of all vertices of $\mathfrak{T}$ whose labels are not in $\mathcal{P}$} .
\]
Suppose $V_{\mathfrak{T}} \ne \varnothing$. Since $\mathfrak{T}$ is well-founded, $V_{\mathfrak{T}}$ has a minimal element $v_0$ (with respect to the partial order induced by $\mathfrak{T}$). But then:
\begin{itemize}

\item
the children of $v_0$ are not in $V_{\mathfrak{T}}$, so their labels must be in $\mathcal{P}$;

\item
the label of $v_0$ can be obtained from the labels of its children by an application of some rule of $\eACTomega$.

\end{itemize}
This contradicts the fact that $\mathcal{P}$ is closed under the rules of $\eACTomega$. Therefore $V_{\mathfrak{T}} = \varnothing$, and in particular, the goal of
$\mathfrak{T}$ must be in $\mathcal{P}$. Consequently $\mathcal{S} \subseteq \mathcal{P}$.
\end{proof}

The presence of an $\omega$-rule also  indicates that our calculus has certain model-theoretic features. In fact,
this rule is indispensable, which implies that the set of all sequents derivable in our calculus is not computably enumerable.

Notice that the right rule for $\KStar$ is formulated in a non-standard way, by isolating the zero case and imposing a non-emptiness condition on the non-zero one. Such a formulation is equivalent to Palka's one:
$$
\infer[(\yields \KStar)_n,\ n \geqslant 0]
{\Pi_1, \ldots, \Pi_n \yields A^*}
{\Pi_1 \yields A^* & \ldots & \Pi_n \yields A^*}
$$
without any restrictions on $\Pi_i$. Indeed, if the antecedent is non-empty, then so is a least one $\Pi_i$. Thus, $n > 0$. The empty $\Pi_i$'s can be just removed, which makes the rule even stronger, getting rid of useless premises. The case of empty antecedent is captured by our $(\yields\KStar)_0$ axiom. The reason for this change is as follows: in the new formulation, for each sequent there is now only a finite choice of rule applications which can derive it. (In the original formulation, the choice was infinite, since one could add meaningless empty $\Pi_i$'s.)

Also notice that the contraction rule here is presented in its {\em non-local} form, allowing contraction of distant instances of ${!}^c A$. In absence of permutation (that is, $c \in \Cc$, but $c \notin \Ec$) this is crucial for cut elimination~\cite{KKNS-mscs-2018}.




\vskip 5pt
The rest of the article is organized as follows. In Section~\ref{sec-prel} we introduce some computational machinery we are going to use further.
In Section~\ref{sec-cutelim} we prove cut elimination in ${!}\mathrm{ACT}_\omega$. In Section~\ref{sec-comp} we study the complexity aspects of
$!\mathrm{ACT}_\omega$, and show the following:
\begin{enumerate}
 \item in the case of $\Cc \ne \varnothing$, the derivability problem in $!\mathrm{ACT}_\omega$ is  $\Pi_1^1$-complete and the closure ordinal of the corresponding derivability operator is $\omega_1^{\mathrm{CK}}$;
 \item in the case of $\Cc = \varnothing$ (that is, no subexponential allows contraction) the closure ordinal is bounded by $\omega^\omega$ and the derivability problem is $\Pi_1^0$-complete. (In other words, with $\Cc = \varnothing$ complexity is the same as for $\mathrm{ACT}_\omega$ without subexponentials.)
\end{enumerate}
Section~\ref{sec-conclusion} concludes the article stating some problems left for future research.

Notice that in the results for $\Cc \ne \varnothing$ cut elimination is used not for the upper bound (proof search), since  such high complexity upper bounds ($\Pi_1^1$ and $\omega_1^{\mathrm{CK}}$) can be as well obtained with cut in the system. Cut elimination is needed for the reduction used to prove the lower bound.

For proving the $\Pi_1^0$ upper bound under the $\Cc = \varnothing$ condition, we use a new, more robust approach, which is less dependent on the concrete structure of the calculus, if compared to the ones by Palka~\cite{Palka-2007} and Das and Pous~\cite{DasPous-2018}.


\section{Some computational machinery} \label{sec-prel}

In this section we introduce some machinery needed for our complexity estimations in Section~\ref{sec-comp}. In particular, we consider {\em second-order}
arithmetic, built on top of second-order predicate logic. The logical symbols used here are kept different from the ones used in ${!}\mathrm{ACT}_\omega$: $\wedge$, $\vee$, $\to$, compared to $\aconj$, $\adisj$, $\BS$. This excludes confusion between the theory studied (${!}\mathrm{ACT}_\omega$) and the
meta-theory employed (second-order arithmetic).

\subsection{$\mathbf{\Pi}^1_1$-sets and $\mathbf{\Delta}^1_1$-sets} \label{subsec-msoa}

Let $\sigma$ be one's favourite signature of Peano arithmetic (say, $\left\{ 0, \mathsf{s}, +, \,\cdot\,, =, \leqslant \right\}$), and let $\mathfrak{N}$ be its \emph{standard model}.\footnote{As
far as degrees of undecidability are concerned, it makes no difference which signature we choose. For our purposes, it will be convenient to think of $\sigma$
as containing symbols for all (total) computable functions as well as all computable relations.}
Throughout the paper we assume the following:
\begin{itemize}

\item the \emph{connective} symbols are $\neg$, $\wedge$ and $\vee$;

\item the \emph{quantifier} symbols are $\forall$ and $\exists$.

\end{itemize}
For our present purposes, it is convenient to treat $\rightarrow$ as defined, rather than as primitive. Next, we restrict our attention to \emph{monadic}
second-order arithmetic~--- bearing in mind, however, that first-order arithmetic allows us to code elements of $\mathbb{N}^{n}$ as elements of
$\mathbb{N}$. Recall, its language $\mathcal{L}_2$ includes two different sorts of variables, namely:
\begin{itemize}

\item \emph{individual variables} $x$, $y$, \dots\ (intended to range over natural numbers);

\item \emph{set variables} $X$, $Y$, \dots\ (intended to range over sets of natural numbers).

\end{itemize}
Accordingly one must distinguish between \emph{individual} and \emph{set quantifiers}, viz.\
\[
{\forall x},\ {\exists x},\ {\forall y},\ {\exists y},\ \dots
\quad \text{and} \quad
{\forall X},\ {\exists X},\ {\forall Y},\ {\exists Y},\ \dots
\]
The \emph{$\mathcal{L}_2$-formulas}~--- or \emph{monadic second-order $\sigma$-formulas}~--- are built up from the first-order~$\sigma$-for\-mu\-las
and the expressions of the form $t \in X$, where $t$ is a $\sigma$-term and $X$ is a set variable, by means of the connective symbols and the
quantifiers in the usual way. As one would expect, we write $\Phi \rightarrow \Psi$ as shorthand for ${\neg \Phi} \vee \Psi$, and 
$\Phi \leftrightarrow \Psi$ for $({\neg \Phi} \vee \Psi) \wedge ({\neg \Psi} \vee \Phi)$. Let
\[
{\mathrm{Th}_2 \left( \mathfrak{N} \right)}\ :=\ \text{the collection of all} ~ \mathcal{L}_2 \text{-sentences true in} ~ \mathfrak{N} .
\]
So $\mathrm{Th}_2 \left( \mathfrak{N} \right)$ denotes the $\mathcal{L}_2$-theory of $\mathfrak{N}$, often called \emph{complete second-order
arithmetic}.

We say an $\mathcal{L}_2$-formula is a \emph{$\Pi^1_{1}$-formula} iff it has the form ${\forall X}\, \Psi$ with $X$ a set variable and $\Psi$
containing no set quantifiers. The \emph{$\Sigma^1_{1}$-formulas} are defined in the same way but with $\exists$ in place of $\forall$. Let
\begin{align*}
{\Pi^1_1 \text{-} \mathrm{Th}_2 \left( \mathfrak{N} \right)}\ &:=\ \text{the collection of all} ~ \Pi^1_1 \text{-sentences true in} ~ \mathfrak{N} ,\\
{\Sigma^1_1 \text{-} \mathrm{Th}_2 \left( \mathfrak{N} \right)}\ &:=\ \text{the collection of all} ~ \Sigma^1_1 \text{-sentences true in} ~ \mathfrak{N} .
\end{align*}
So $\Pi^1_1 \text{-} \mathrm{Th}_2 \left( \mathfrak{N} \right)$ and $\Sigma^1_1 \text{-} \mathrm{Th}_2 \left( \mathfrak{N} \right)$ denote respectively
the $\Pi^1_1$- and $\Sigma^1_1$-fragments of $\mathrm{Th}_2 \left( \mathfrak{N} \right)$, often called its \emph{universal} and \emph{existential
fragments}.

Let $P, Q \subseteq \mathbb{N}$. We say $P$ is \emph{$m$-reducible} to $Q$, written $P \leqslant_m Q$, iff there exists a total~com\-p\-u\-tab\-le
function $f$ from $\mathbb{N}$ to $\mathbb{N}$ such that $f^{-1} \left[ Q \right] = P$, i.e.\ for every $k \in \mathbb{N}$,
\[
{k\ \in\ P}
\quad \Longleftrightarrow \quad
{{f \left( k \right)}\ \in\ Q} .
\]
Then $P$ and $Q$ are called \emph{$m$-equivalent}, written $P \equiv_m Q$, iff they are $m$-reducible to each other. Also, we call $P$ and $Q$
\emph{computably isomorphic}, written $P \simeq Q$, iff there is a one-one total~com\-p\-u\-tab\-le function $f$ from $\mathbb{N}$ onto $\mathbb{N}$
such that $f \left[ P \right] = Q$.\footnote{In
effect, in many important special cases $P \equiv_m Q$ implies $P \simeq Q$.}
Now $P$ is called:
\begin{itemize}

\item \emph{$\Pi^1_1$-bounded} iff there exists a $\Pi^1_1$-formula $\Phi \left( x \right)$ defining $P$ in $\mathfrak{N}$;

\item \emph{$\Pi^1_1$-hard} iff for any $\Pi^1_1$-bounded $R \subseteq \mathbb{N}$ we have $R \leqslant_m P$;

\item \emph{$\Pi^1_1$-complete} iff it is $\Pi^1_1$-bounded and $\Pi^1_1$-hard.

\end{itemize}
Similarly with $\Sigma^1_1$ in place of $\Pi^1_1$. Evidently for every $P \subseteq \mathbb{N}$ the following hold:
\begin{align*}
P ~ \text{is} ~ \Pi^1_1 \text{-bounded} \quad &\Longleftrightarrow \quad
\overline{P} ~ \text{is} ~ \Sigma^1_1 \text{-bounded} ;\\
P ~ \text{is} ~ \Pi^1_1 \text{-hard} \quad &\Longleftrightarrow \quad
\overline{P} ~ \text{is} ~ \Sigma^1_1 \text{-hard} ;\\
P ~ \text{is} ~ \Pi^1_1 \text{-complete} \quad &\Longleftrightarrow \quad
\overline{P} ~ \text{is} ~ \Sigma^1_1 \text{-complete} .
\end{align*}
Thus, without loss of generality, we may concentrate on $\Pi^1_1$. Further~--- to simplify matters, for each formal language which appears in this
article, we shall tacitly fix a G\"{o}del numbering of its objects, and occasionally identify its formulas, etc.\ with their G\"{o}del numbers.

\begin{Folk} \label{folk-msoa-1} \renewcommand{\theenumi}{\alph{enumi}}
Let $P \subseteq \mathbb{N}$. Then:
\begin{enumerate}

\item $P$ is $\Pi^1_1$-bounded iff $P \leqslant_m \Pi^1_1 \text{-} \mathrm{Th}_2 \left( \mathfrak{N} \right)$;

\item $P$ is $\Pi^1_1$-hard iff $\Pi^1_1 \text{-} \mathrm{Th}_2 \left( \mathfrak{N} \right) \leqslant_m P$;

\item $P$ is $\Pi^1_1$-complete iff $P \equiv_m \Pi^1_1 \text{-} \mathrm{Th}_2 \left( \mathfrak{N} \right)$ iff
$P \simeq \Pi^1_1 \text{-} \mathrm{Th}_2 \left( \mathfrak{N} \right)$.

\end{enumerate}
Similarly with $\Sigma^1_1$ in place of $\Pi^1_1$.
\end{Folk}

This is, in fact, intimately connected with:

\begin{Folk} \label{folk-msoa-2}
There exists a $\Pi^1_1$-formula $\Pi^1_1 \text{-} \mathrm{SAT} \left( x, y \right)$ such that for any $\Pi^1_1$-formula $\Phi \left( x \right)$,
\[
{\mathfrak{N} \vDash {\forall x}\, \left( {\Pi^1_1 \text{-} \mathrm{SAT}} \left( x, {\# \Phi} \right) \leftrightarrow \Phi \left( x \right) \right)}
\]
where $\# \Phi$ denotes the G\"{o}del number of $\Phi$. Similarly with $\Sigma^1_1$ in place of $\Pi^1_1$.
\end{Folk}

Finally, call $P$ \emph{hyperarithmetical}, or \emph{$\Delta^1_1$-bounded}, iff it is both $\Pi^1_1$-bounded and $\Sigma^1_1$-boun\-ded. By analogy
with what happened earlier, one can define what it means for $P$ to be \emph{$\Delta^1_1$-hard} and \emph{$\Delta^1_1$-com\-p\-le\-te}, but the
latter notion turns out to be devoid of content:

\begin{Folk} \label{folk-msoa-3}
There exists no $\Delta^1_1$-complete set.
\end{Folk}

Consequently, a hyperarithmetical set cannot be $\Pi^1_1$-hard, and cannot be $\Sigma^1_1$-hard~--- since otherwise it would be $\Delta^1_1$-complete. The reader might
consult \cite{Rogers-1967} for more information.


\subsection{Kleene's $\mathcal{O}$} \label{subsec-o}

To simplify the discussion, let
\begin{align*}
\mathsf{Ord}\ &:=\ \textrm{the class of all ordinals} ,\\
{\mathsf{L} \text{-} \mathsf{Ord}}\ &:=\ \textrm{the class of all limit ordinals} ,\\
{\mathsf{C} \text{-} \mathsf{Ord}}\ &:=\ \textrm{the class of all constructive ordinals} .
\end{align*}
The least element of $\mathsf{Ord} \setminus {\mathsf{C} \text{-} \mathsf{Ord}}$ is traditionally called the \emph{Church--Kleene ordinal}, and denoted
by $\omega_1^\mathrm{CK}$. \emph{Kleene's system of notation for $\mathsf{C} \text{-} \mathsf{Ord}$} consists of:
\begin{itemize}

\item a special partial function $\nu_\mathcal{O}$ from $\mathbb{N}$ onto $\mathsf{C} \text{-} \mathsf{Ord}$;

\item a special ordering relation $<_\mathcal{O}$ on $\mathrm{dom} \left( \nu_\mathcal{O} \right)$ which mimics $<$ on $\mathsf{C} \text{-}
\mathsf{Ord}$.\footnote{As
usual, if $f$ is a partial function, we write $\mathrm{dom} \left( f \right)$ for its domain.}

\end{itemize}
We say $n \in \mathbb{N}$ is a \emph{notation for $\alpha \in \mathsf{C} \text{-} \mathsf{Ord}$} iff $\nu_\mathcal{O} \left( n \right) = \alpha$. Using
one's favourite universal partial computable (two-place) function $\ae$, $\nu_\mathcal{O}$ and $<_\mathcal{O}$ are defined simultaneously by induction:
\begin{itemize}

\item The ordinal $0$ receives the only notation, namely $1$. Thus $\nu_\mathcal{O}^{-1} \left( 0 \right) = \left\{ 1 \right\}$.

\item Suppose all ordinals below $\alpha$ have received their notations, and assume that $<_\mathcal{O}$ has been defined on these notations.

\begin{itemize}

\item If $\alpha = \beta + 1$, then $\alpha$ receives the notations $\left\{ 2^k \mid k \in \nu_\mathcal{O}^{-1} \left( \beta \right) \right\}$.
Furthermore, for each $k \in \nu_\mathcal{O}^{-1} \left( \beta \right)$ we set $i <_\mathcal{O} 2^k$ iff $i = k$ or $i <_\mathcal{O} k$.

\item If $\alpha \in \mathsf{L} \text{-} \mathsf{Ord}$, then $\alpha$ receives the notation $3 \times 5^k$ for any $k$ such that
\[
{\ae_k \left( 0 \right)}\ <_\mathcal{O}\ {\ae_k \left( 1 \right)}\ <_\mathcal{O}\ {\ae_k \left( 2 \right)}\ <_\mathcal{O}\ \dots
\quad \text{and} \quad
{\bigcup\nolimits_{i \in \mathbb{N}} \nu_\mathcal{O} \left( \ae_k \left( i \right) \right)}\ = \alpha
\]
(hence $\ae_k$ must be total, and all $\ae_k \left( i \right)$ must be elements of $\bigcup\nolimits_{\beta < \alpha} \nu_\mathcal{O}^{-1} \left( \beta
\right)$). Furthermore, for each such $k$ we set $i <_\mathcal{O} 3 \times 5^k$ iff $i <_\mathcal{O} \ae_k \left( j \right)$ for some $j$.

\end{itemize}

\end{itemize} 
In what follows we shall often write $n \in \mathcal{O}$ instead of $n \in \mathrm{dom} \left( \nu_\mathcal{O} \right)$. It turns out that $\mathrm{dom}
\left( \nu_\mathcal{O} \right)$ has the same complexity as the universal fragment of complete second-order arithmetic:

\begin{Folk}\label{folk-o-1} 
$\mathrm{dom} \left( \nu_\mathcal{O} \right)$ is $\Pi^1_1$-complete.
\end{Folk}

Moreover, the restriction of $<_\mathcal{O}$ to $\left\{ k \mid k <_\mathcal{O} n \right\}$ is computably enumerable~uniformly in $n$:

\begin{Folk} \label{folk-o-2} 
There exists a computable $f: \mathbb{N} \rightarrow \mathbb{N}$ such that for all $n \in \mathcal{O}$,
\[
{\mathrm{dom} \left( \ae_{f \left( n \right)} \right)}\ =\ {\left\{ k \mid k <_\mathcal{O} n \right\}} .
\] \vspace{-5mm}
\end{Folk}

Readers who want to know more about constructive ordinals and systems of notation might con\-sult~\cite{Rogers-1967} or \cite{Sacks-1990}.


\subsection{Inductive definitions} \label{subsec-id}

A function $F$ from $\mathcal{P} \left( \mathbb{N} \right)$ to $\mathcal{P} \left( \mathbb{N} \right)$ is said to be \emph{monotone} iff for all
$P, Q \subseteq \mathbb{N}$,
\[
P\ \subseteq\ Q \quad \Longrightarrow \quad
{F \left( P \right)}\ \subseteq\ {F \left( Q \right)} .
\]
Given such an $F$, for each $S \subseteq \mathbb{N}$ we inductively define
\[
{F^{\alpha} \left( S \right)}\ :=\
\begin{cases}
S                                                       &\text{if} ~\, \alpha = 0 ,\\
{F \left( F^{\beta} \left( S \right) \right)}           &\text{if} ~\, \alpha = \beta +1 ,\\
{\bigcup_{\beta < \alpha} {F^{\beta} \left( S \right)}} &\text{if} ~\, \alpha \in {\mathsf{L} \text{-} \mathsf{Ord}} \setminus \left\{ 0 \right\} .
\end{cases}
\]
Evidently the resulting transfinite sequence is non-decreasing, viz.\ for any $\alpha, \beta \in \mathsf{Ord}$,
\[
\alpha\ <\ \beta \quad \Longrightarrow \quad
{F^{\alpha} \left( S \right)}\ \subseteq\ {F^{\beta} \left( S \right)} .
\]
Furthermore, it stabilises, by a version of the well-known Knaster--Tarski theorem:

\begin{Folk} \label{folk-id-1}
Let $F: \mathcal{P} \left( \mathbb{N} \right) \rightarrow \mathcal{P} \left( \mathbb{N} \right)$ be monotone. Then for every $S \subseteq \mathbb{N}$
there exists a least $\alpha \in \mathsf{Ord}$ such that $F^{\alpha+1} \left( S \right) = F^{\alpha} \left( S \right)$~--- so $F^{\alpha} \left( S \right)$ is
the least fixed point of $F$ containing $S$.
\end{Folk}

For every monotone $F: \mathcal{P} \left( \mathbb{N} \right) \rightarrow \mathcal{P} \left( \mathbb{N} \right)$ the least $\alpha \in \mathsf{Ord}$
which satisfies $F^{\alpha+1} \left( \varnothing \right) = F^{\alpha} \left( \varnothing \right)$ is called the \emph{closure ordinal} of $F$. Next,
for each $\mathcal{L}_2$-formula $\Phi \left( x, X \right)$ we define the function $\left[ \Phi \right]$ from $\mathcal{P} \left( \mathbb{N} \right)$
to $\mathcal{P} \left( \mathbb{N} \right)$ as follows:
\[
{\left[ \Phi \right] \left( P \right)}\ :=\ {\left\{ n \in \mathbb{N} \mid \mathfrak{N} \models \Phi \left( n, P \right) \right\}} .
\]
We say $F: \mathcal{P} \left( \mathbb{N} \right) \rightarrow \mathcal{P} \left( \mathbb{N} \right)$ is a \emph{$\Pi^1_1$-operator} iff $F = \left[ \Phi
\right]$ for some $\Pi^1_1$-formula $\Phi$.\footnote{Similarly
for $\Sigma^1_1$. However, the properties of $\Sigma^1_1$-operators are quite different from those of $\Pi^1_1$-operators; see \cite{Moschovakis-1974}
and \cite{Hinman-1978} for more information.}

\begin{Folk} \label{folk-id-2} \renewcommand{\theenumi}{\alph{enumi}}
Let $F: \mathcal{P} \left( \mathbb{N} \right) \rightarrow \mathcal{P} \left( \mathbb{N} \right)$ be a monotone $\Pi^1_1$-operator. Then:
\begin{enumerate}

\item the least fixed point of $F$ is $\Pi^1_1$-bounded;

\item the closure ordinal of $F$ is less than or equal to $\omega_1^{\mathrm{CK}}$.

\end{enumerate}
\end{Folk}

\smallskip
We call $F$ a \emph{hyperarithmetical operator}, or a \emph{$\Delta^1_1$-operator}, iff there are a $\Pi^1_1$-formula $\Phi$ and a $\Sigma^1_1$-formula
$\Psi$ such that
\[
{\left[ \Phi \right]} \ =\ {F}\ =\ {\left[ \Psi \right]} .
\]
Further~--- $F$ is said to be an \emph{arithmetical} (or \emph{elementary}) \emph{operator} iff $F = \left[ \Phi \right]$ for some
$\mathcal{L}_2$-for\-mu\-la~$\Phi$~with no set quantifiers;
also, $\mathcal{L}_2$-formulas with no set quantifiers are traditionally called \emph{arithmetical} (or \emph{elementary}).
For discussion, examples, and related results, the reader might consult \cite{Moschovakis-1974} and \cite{Hinman-1978}.


We shall use a specific operator $F$, namely, the {\em immediate derivability operator of ${!}\mathrm{ACT}_\omega$,} denoted by $\Dcut$. Let $\mathrm{Seq}$ be the set of all sequents in the language of $!\mathrm{ACT}_\omega$ and let $\Dcut$ be the function from $\mathcal{P}(\mathrm{Seq})$ to $\mathcal{P}(\mathrm{Seq})$ such that for any $S \subseteq \mathrm{Seq}$ and $s \in \mathrm{Seq}$,
\[
s\ \in\ {\Dcut \left( S \right)}
\quad \Longleftrightarrow \quad
\begin{array}{c}
s ~ \text{is an element of} ~ S ~ \text{or} ~ s ~ \text{can be obtained from}\vspace{0.5mm}\\
\text{elements of} ~ S ~ \text{by one application of some rule of} ~ {! \mathrm{ACT}_{\omega}} .
\end{array}
\]
Here axioms of $!\mathrm{ACT}_\omega$ are considered as rules with zero premises; so axioms belong to
$\Dcut(S)$ for each $S$. Finally, remember from section~\ref{sec-calculus} that the collection of all sequents derivable in $!\mathrm{ACT}_\omega$ coincides with the smallest set of sequents closed under the rules of $!\mathrm{ACT}_\omega$; thus this collection must be the least fixed point of $\Dcut$.



\section{Cut elimination in ${!}\mathbf{ACT}_\omega$}
\label{sec-cutelim}

In this section we prove that any sequent provable in $!\mathrm{ACT}_\omega$ can be proved without using the cut rule.
This proof is a juxtaposition of Palka's cut elimination proof for infinitary action logic~\cite{Palka-2007} and the cut-elimination proof for non-commutative linear logic with subexponentials~\cite{KKNS-mscs-2018}.

First we show how to eliminate one cut. Let $\Dnocut$ denote the {\em immediate derivability operator of $!\mathrm{ACT}_\omega$
without the cut rule} (in $\Dcut$, cut is allowed). Purely for exposition, if $\alpha$ is an ordinal, we shall often write $\hat{\mathcal{S}}_\alpha$
instead of $\Dnocut^\alpha (\varnothing)$.\footnote{Here
$\Dnocut$ plays the role of $F$ from the previous section.}


Cut elimination will be proved by transfinite induction. Let us define the parameters used in this inductive argument. The {\em complexity} of a formula is defined in a traditional way, as the total number of connective occurrences. For a sequent $s$ derivable without cut, let its {\em rank} be the smallest $\alpha$ such that $s \in \hat{\mathcal{S}}_\alpha$. These ranks are always successor ordinals:
\begin{quote}
Let $\alpha$ be the rank of $s$. Assume that $\alpha \in \mathsf{L} \text{-} \mathsf{Ord}$. Then $\alpha \ne 0$ (because $\hat{\mathcal{S}}_0 = \varnothing$) and $\hat{\mathcal{S}}_\alpha = \bigcup\nolimits_{\beta < \alpha} \hat{\mathcal{S}}_\beta$, hence $s \in \hat{\mathcal{S}}_\beta$ for some $\beta < \alpha$ --- which contradicts the choice of $\alpha$. Thus $\alpha$ is not limit, so $\alpha$ has the form $\beta + 1$.
\end{quote}
An important observation is that a sequent belongs to $\hat{\mathcal{S}}_{\alpha+1}$ iff it can be obtained from elements of $\hat{\mathcal{S}}_\alpha$ by one application of some rule. (In particular, sequents from $\hat{\mathcal{S}}_1$ are simply axioms.)

As an example, consider the application
\[
\infer
{{\mathbf{1}}^{\ast} \yields \mathbf{1}}
{{\left( {\mathbf{1}}^n \yields \mathbf{1} \right)}_{n \in \NN}}
\]
of the $\omega$-rule. Notice that for each $n \in \mathbb{N}$ the sequent ${\mathbf{1}}^n \yields \mathbf{1}$ has rank $n$. Therefore $\hat{\mathcal{S}}_{\alpha}$ includes $\left\{ {\mathbf{1}}^n \yields \mathbf{1} \mid n \in \mathbb{N} \right\}$ iff $\alpha \geqslant \omega$. Thus the rank of ${\mathbf{1}}^{\ast} \yields \mathbf{1}$ must be $\omega + 1$.

\begin{Th}\label{Th:onecut}
 If $\Pi \yields A$ and $\Gamma, A, \Delta \yields B$ are derivable without using cut, then so is $\Gamma,\Pi,\Delta \yields B$.
\end{Th}

\begin{proof}

In the presence of contraction, attempts to establish cut elimination by induction fail when the inductive argument comes across applications of contraction. Following the classical strategy of Gentzen~\cite{Gentzen-1935}, we introduce the {\em mix} rule, which is a combination of contraction (in our case, non-local contraction) and cut:
$$
\infer[(\mathrm{mix}),\ c \in \Cc]
{\Delta_0, \Delta_1, \ldots, \Delta_i, \Pi, \Delta_{i+1}, \ldots,
\Delta_n \yields B}
{\Pi \yields {!}^c A & \Delta_0, {!}^c A, \Delta_1, {!}^c A, \ldots,
\Delta_{i}, {!}^c A, \Delta_{i+1}, {!}^c A, \ldots, {!}^c A, \Delta_n \yields B}
$$

Notice that the mix rule is available only for formulae of the form ${!}^c A$, where $c \in C$. Thus, unlike the intuitionistic situation, mix is {\em not} a generalization of cut. We shall perform cut and mix elimination by joint induction. For mix, however, we shall consider only a specific case when its left premise was introduced by the $(\yields {!})$ rule, as this will be sufficient for eliminating cut.\footnote{This simplification of the proof was suggested by one of the referees.}

Namely, we are going to prove the conjunction of the following two claims:
\begin{enumerate}
 \item if $\Pi \yields A$ and $\Gamma, A, \Delta \yields B$ are derivable without using cut and mix, then so is $\Gamma, \Pi, \Delta \yields B$;
 \item if $\Pi \yields {!}^c A$ and $\Delta_0, {!}^c A, \Delta_1, {!}^c A, \ldots,
\Delta_{i}, {!}^c A, \Delta_{i+1}, {!}^c A, \ldots, {!}^c A, \Delta_n \yields B$ are derivable without using cut and mix and, moreover, $\Pi \yields {!}^c A$ is derived from $\Pi \yields A$ by $(\yields {!})$, then $\Delta_0, \Delta_1, \ldots, \Delta_i, \Pi, \Delta_{i+1}, \ldots,
\Delta_n \yields B$ is also derivable without using cut and mix.
\end{enumerate}

We proceed by nested induction on the following parameters:
\begin{enumerate}
 \item complexity of the formula being cut ($A$ for cut, ${!}^c A$ for mix), measured just as the total number of variable, constant, and connective occurrences;
 \item rank of the left premise ($\Pi \yields A$), only for cut;
 \item rank of the right premise.
\end{enumerate}

\subsubsection*{Cut Elimination}

Let the rank of $\Pi \yields A$ be $\alpha + 1$. As noticed above, $\Pi \yields A$ can be obtained by an application of a rule from sequent(s) of rank $\leqslant \alpha$ (in particular, for $\alpha = 0$ the sequent should be an axiom).

Consider the possible cases:

\vskip 2pt
{\em Case 1 (axiom).} $\Pi \yields A$ is the $(\mathrm{id})$ axiom, {\em i.e.,} $A \yields A$. Then the right premise, $\Gamma, A, \Delta \yields B$, coincides with the goal sequent, nothing to prove (cut disappears).

\vskip 2pt
{\em Case 2 (non-principal).} $\Pi \yields A$ is obtained by a rule operating in  the left-hand side of the sequent. Let us denote this rule by $\mathrm{R}$.
If $\mathrm{R}$ is $(\mconj\yields)$, $(\mathbf{1}\yields)$, $(\adisj\yields)$, $(\aconj\yields)$, $(\KStar\yields)_\omega$, or any of the rules operating ${!}^s$, then it has one or several premises of the form $\widetilde{\Pi} \yields A$. These premises have rank $\leqslant \alpha$, and by induction (complexity of $A$ unchanged, rank of the left premise reduced) we get derivability of $\Gamma, \widetilde{\Pi}, \Delta \yields B$ without cut and mix. The rule $\mathrm{R}$ is also applicable in the $\Gamma, \ldots, \Delta$ context. Thus, we get cut-free derivability of $\Gamma, \Pi, \Delta \yields B$.

The case of $(\SL\yields)$ or $(\BS\yields)$ is similar. In this case we proceed by induction with the right premise of this rule and then apply the rule:
$$
\infer[(\CUT)]{\Gamma, \Pi_1, \Pi_2, E \BS F, \Pi_3, \Delta \yields B}{\infer[(\BS\yields)]{\Pi_1, \Pi_2, E \BS F, \Pi_3 \yields A}{\Pi_2 \yields E & \Pi_1, F, \Pi_3 \yields A} &
\Gamma, A, \Delta \yields B}
$$
transforms into
$$
\infer[(\BS\yields)]{\Gamma, \Pi_1, \Pi_2, E \BS F, \Pi_3, \Delta \yields B}
{\Pi_2 \yields E & \infer[(\CUT)]
{\Gamma, \Pi_1, F, \Pi_3, \Delta \yields B}
{\Pi_1, F, \Pi_3 \yields A & \Gamma, A, \Delta \yields B}}
$$

{\em Case 3 (left principal).} $\Pi \yields A$ is obtained by a rule introducing the main connective of $A$. Let $\beta + 1$ be the rank of $\Gamma, A, \Delta \yields B$. Consider several subcases.

{\em Subcase 3.1 (right axiom).} $\Gamma, A, \Delta \yields B$ is the $(\mathrm{id})$ axiom, $A \yields A$. Cut disappears.

{\em Subcase 3.2 (right non-principal).} $\Gamma, A, \Delta \yields B$ is obtained by a rule $\mathrm{R}$ which does not change $A$.
This case is considered similarly to Case 2. Indeed, if $\mathrm{R}$ is $(\mconj\yields)$, $(\mathbf{1}\yields)$, $(\adisj\yields)$, $(\aconj\yields)$, $(\KStar\yields)_\omega$, $(\yields\SL)$, $(\yields\BS)$, $(\yields\aconj)$, $(\yields\adisj)$, $({!}\yields)$, $(\PERM)$, $(\NCONTR)$, or $(\WEAK)$, then it derives $\Gamma, A, \Delta \yields B$ from one or several premises of the form $\widetilde{\Gamma}, A, \widetilde{\Delta} \yields \widetilde{B}$.
By induction (on the third parameter), we obtain cut-free derivability of $\widetilde{\Gamma}, \Pi, \widetilde{\Delta} \yields \widetilde{B}$, and then apply $\mathrm{R}$.

For ``branching'' rules $(\SL\yields)$, $(\BS\yields)$, $(\yields\mconj)$, and $(\yields\KStar)_n$, the cut formula $A$ goes to one of the premises. For this premise, we proceed by induction, and afterwards apply $\mathrm{R}$. Let us show this on $(\yields\KStar)_n$; other cases for $(R)$ are more standard and considered in~\cite{KKNS-mscs-2018}.
$$
\infer[(\CUT)]{\Delta_1, \ldots, \Delta_{i-1}, \Delta'_i, \Pi, \Delta''_i, \Delta_{i+1}, \ldots, \Delta_n \yields B^*}
{\Pi \yields A & \infer[(\yields\KStar)_n]{\Delta_1, \ldots, \Delta_{i-1}, A, \Delta''_i, \Delta_{i+1}, \ldots, \Delta_n \yields B^*}{\Delta_1 \yields B & \ldots & \Delta_{i-1} \yields B & \Delta'_i A, \Delta''_i \yields B & \Delta_{i+1} \yields B & \ldots & \Delta_n \yields B}}
$$
transforms into
$$
\infer[(\yields\KStar)_n]{\Delta_1, \ldots, \Delta_{i-1}, \Delta'_i, \Pi, \Delta''_i, \Delta_{i+1}, \ldots, \Delta_n \yields B^*}
{\Delta_1 \yields B & \ldots & \Delta_{i-1} \yields B &
\infer[(\CUT)]{\Delta'_i, \Pi, \Delta''_i \yields B}
{\Pi \yields A & \Delta'_i, A, \Delta''_i \yields B} &
\Delta_{i+1} \yields B & \ldots & \Delta_n \yields B}
$$

Finally, the case of $(\yields {!})$ is specific. In this case we use the fact that $A = {!}^r A'$ (for some $r$) and $\Pi \yields A$ was also derived using $(\yields {!})$ (left principality). Thus, $\Pi = {!}^{r_1} C_1, \ldots, {!}^{r_k} C_k$, and the whole situation is as follows:
$$
\infer[(\CUT)]{{!}^{s_1} E_1, \ldots, {!}^{s_{i-1}} E_{i-1},
{!}^{r_1} C_1, \ldots, {!}^{r_k} C_k, {!}^{s_{i+1}} E_{i+1}, \ldots {!}^{s_n} E_n \yields {!}^s B'}
{{!}^{r_1} C_1, \ldots, {!}^{r_k} C_k \yields {!}^r A' &
\infer[(\yields{!})]{{!}^{s_1} E_1, \ldots, {!}^{s_{i-1}} E_{i-1}, {!}^r A',
{!}^{s_{i+1}} E_{i+1}, \ldots {!}^{s_n} E_n \yields {!}^s B'}
{{!}^{s_1} E_1, \ldots, {!}^{s_{i-1}} E_{i-1}, {!}^r A',
{!}^{s_{i+1}} E_{i+1}, \ldots {!}^{s_n} E_n \yields B'}}
$$
Here cut gets propagated through $(\yields {!})$ as follows:
$$
\infer[(\yields {!})] {{!}^{s_1} E_1, \ldots, {!}^{s_{i-1}} E_{i-1},
{!}^{r_1} C_1, \ldots, {!}^{r_k} C_k, {!}^{s_{i+1}} E_{i+1}, \ldots {!}^{s_n} E_n \yields {!}^s B'}
{\infer[(\CUT)]{{!}^{s_1} E_1, \ldots, {!}^{s_{i-1}} E_{i-1},
{!}^{r_1} C_1, \ldots, {!}^{r_k} C_k, {!}^{s_{i+1}} E_{i+1}, \ldots {!}^{s_n} E_n \yields  B'}
{{!}^{r_1} C_1, \ldots, {!}^{r_k} C_k \yields {!}^r A' &
{!}^{s_1} E_1, \ldots, {!}^{s_{i-1}} E_{i-1}, {!}^r A',
{!}^{s_{i+1}} E_{i+1}, \ldots {!}^{s_n} E_n \yields B'}}
$$
The new application of $(\yields {!})$ is legal due to transitivity of the preorder on subexponential labels:
$r_j \succcurlyeq r \succcurlyeq s$. (Here the first inequation is due to the fact that the left premise of cut was introduced by $(\yields{!})$.)

{\em Subcase 3.3 (principal vs.\ principal).}
$\Gamma, A, \Delta \yields B$ is obtained by the left logical rule which introduces the main connective of $A$. In this subcase cut for $A$ gets reduced to cuts for its subformulae, which are eliminated by induction on the first parameter (complexity of the formula being cut).

Consider possible situations depending on the main connective of $A$.
\begin{enumerate}
 \item $A = A_1 \SL A_2$:
 $$
 \infer[(\CUT)]
 {\Gamma, \Pi, \Upsilon, \Delta \yields B}
 {\infer[(\yields\SL)]{\Pi \yields A_1 \SL A_2}
 {\Pi, A_2 \yields A_1} &
 \infer[(\SL\yields)]{\Gamma, A_1 \SL A_2, \Upsilon,
 \Delta \yields B}{\Upsilon \yields A_2 & \Gamma, A_1, \Delta \yields B}}
 $$
 gets transformed into
 $$
 \infer[(\CUT)]
 {\Gamma, \Pi, \Upsilon, \Delta \yields B}
 {\Upsilon \yields A_2 & \infer[(\CUT)]{\Gamma, \Pi, A_2, \Delta \yields B}{\Pi, A_2 \yields A_1 & \Gamma, A_1, \Delta \yields B}}
 $$
 Here both cuts have smaller complexity of the formula being cut. Thus, we apply induction hypothesis and first establish cut-free derivability of $\Gamma, \Pi, A_2, \Delta \yields B$, then of $\Gamma, \Pi, \Phi, \Delta \yields B$.
 \item $A = A_2 \BS A_1$. Symmetric.
 \item $A = A_1 \mconj A_2$:
 $$
 \infer[(\CUT)]
 {\Gamma, \Pi_1, \Pi_2, \Delta \yields B}
 {\infer[(\yields\mconj)]{\Pi_1, \Pi_2 \yields A_1 \mconj A_2}
 {\Pi_1 \yields A_1 & \Pi_2 \yields A_2} &
 \infer[(\mconj\yields)]{\Gamma, A_1 \mconj A_2, \Delta \yields B}{\Gamma, A_1, A_2, \Delta \yields B}}
 $$
 gets transformed into
 $$
 \infer[(\CUT)]
 {\Gamma, \Pi_1, \Pi_2, \Delta \yields B}
 {\Pi_1 \yields A_1 & \infer[(\CUT)]
 {\Gamma, A_1, \Pi_2, \Delta \yields B}
 {\Pi_2 \yields A_2 &
 \Gamma, A_1, A_2, \Delta \yields B}}
 $$
 Again, complexity of the formula being cut gets reduced.
 \item $A = A_1 \adisj A_2$:
 $$
 \infer[(\CUT)]{\Gamma, \Pi, \Delta \yields B}
 {\infer[(\yields\adisj)]{\Pi \yields A_1 \adisj A_2}
 {\Pi \yields A_i} &
 \infer[(\adisj\yields)]{\Gamma, A_1 \adisj A_2, \Delta \yields B}{\Gamma, A_1, \Delta \yields B &
 \Gamma, A_2, \Delta \yields B}}
 $$
 transforms into
 $$
 \infer[(\CUT)]{\Gamma, \Pi, \Delta \yields B}
 {\Pi \yields A_i & \Gamma, A_i, \Delta \yields B}
 $$
 \item $A = A_1 \aconj A_2$:
 $$
 \infer[(\CUT)]{\Gamma, \Pi, \Delta \yields B}
 {\infer[(\yields\aconj)]{\Pi \yields A_1 \aconj A_2}
 {\Pi \yields A_1 & \Pi \yields A_2} &
 \infer[(\aconj\yields)]{\Gamma, A_1 \aconj A_2, \Delta \yields B}
 {\Gamma, A_j, \Delta \yields B}}
 $$
 transforms into
 $$
 \infer[(\CUT)]{\Gamma, \Pi, \Delta \yields B}
 {\Pi \yields A_j & \Gamma, A_j, \Delta \yields B}
 $$
 \item $A = A_1^*$:
 $$
 \infer[(\CUT)]
 {\Gamma, \Pi_1, \ldots, \Pi_n, \Delta \yields B}
 {\infer[(\yields\KStar)_n]{\Pi_1, \ldots, \Pi_n \yields A_1^*}{\Pi_1 \yields A_1 & \ldots & \Pi_n \yields A_1} &
 \infer[(\KStar\yields)_\omega]
 {\Gamma, A_1^*, \Delta \yields B}
 {\bigl( \Gamma, A_1^k, \Delta \bigr)_{k=0}^{\infty} \yields B}}
 $$
 transforms into
 $$
 \infer[(\CUT),\mbox{ $n$ times}]
 {\Gamma, \Pi_1, \ldots, \Pi_n, \Delta \yields B}
 {\Pi_1 \yields A_1 & \ldots & \Pi_n \yields A_1 &
 \Gamma, A_1^n, \Delta \yields B}
 $$
 (out of the premises of the $\omega$-rule we take the one
 with $k = n$, others get ignored).

 \item $A = {!}^s A'$:
 $$
 \infer[(\CUT)]{\Gamma, \Pi, \Delta \yields B}
 {\infer[(\yields{!})]{\Pi \yields {!}^s A'}{\Pi \yields A'}
 & \infer[({!}\yields)]{\Gamma, {!}^s A', \Delta \yields B}
 {\Gamma, A', \Delta \yields B}}
 $$
 transforms into
 $$
 \infer[(\CUT)]{\Gamma, \Pi, \Delta \yields B}
 {\Pi \yields A' & \Gamma, A', \Delta \yields B}
 $$
\end{enumerate}

{\em Subcase 3.4 (principal vs.\ !-structural).}
$A = {!}^s A'$ and the rule yielding $\Gamma, {!}^s A', \Delta \yields B$ is $(\WEAK)$, $(\PERM)$, or $(\NCONTR)$ operating ${!}^s A'$.

In the $(\WEAK)$ case, cut gets replaced by a series of weakenings:
$$
\infer[(\CUT)]{\Gamma, \Pi, \Delta \yields B}{\Pi \yields {!}^w A' & \infer[(\WEAK),\ w \in \Wc]{\Gamma, {!}^w A', \Delta \yields B}{\Gamma, \Delta \yields B}}
$$
Recall that $\Pi \yields {!}^w A'$ was obtained by $(\yields{!})$, thus, $\Pi = {!}^{s_1} C_1, \ldots, {!}^{s_k} C_k$, and $s_i \succcurlyeq w$ for $i = 1, \ldots, k$. Therefore, $s_i \in \Wc$, and each formula of $\Pi$ can be weakened. Thus, $\Gamma, \Pi, \Delta \yields B$ is derived from $\Gamma, \Delta \yields B$ by applying $(\WEAK)$ $k$ times.

The $(\PERM)$ rule can be exchanged with cut in the following way, reducing the rank of the right premise:
$$
\infer[(\CUT)]{\Gamma, \Pi, \Delta_1, \Delta_2 \yields B}
{\Pi \yields {!}^p A' &
\infer[(\PERM),\ p\in \Ec]{\Gamma, {!}^p A', \Delta_1, \Delta_1 \yields B}
{\Gamma, \Delta_1, {!}^p A', \Delta_2 \yields B}
}
$$
gets replaced with
$$
\infer[(\PERM)\mbox{ several times}]{\Gamma, \Pi, \Delta_1, \Delta_2 \yields B}
{\infer[(\CUT)]{\Gamma, \Delta_1, \Pi, \Delta_2 \yields B}
{\Pi \yields {!}^p A' &
\Gamma, \Delta_1, {!}^p A', \Delta_2 \yields B}}
$$
Since $\Pi \yields {!}^p A'$ is obtained by $(\yields {!})$, $\Pi = {!}^{r_1} C_1, \ldots, {!}^{r_k} C_k$, where $r_i \succcurlyeq p$. Therefore, $r_i \in \Ec$, and permutation rules can be applied to each formula in $\Pi$.

In the $(\NCONTR)$ case, cut gets replaced by mix with a smaller rank of the right premise:
$$ \infer[(\CUT)]{\Gamma, \Pi, \Delta_1, \Delta_2 \yields B}
{\Pi \yields {!}^c A' & \infer[(\NCONTR),\ c \in \Cc]{\Gamma, {!}^c A', \Delta_1, \Delta_2 \yields B}{\Gamma, {!}^c A', \Delta_1, {!}^c A', \Delta_2 \yields B}}$$
transforms into
$$
\infer[(\mathrm{mix}),]{\Gamma, \Pi, \Delta_1, \Delta_2 \yields B}
{\Pi \yields {!}^c A' & \Gamma, {!}^c A', \Delta_1, {!}^c A', \Delta_2 \yields B}
$$
and we use the induction hypothesis of {\em mix elimination} (recall that we proceed by joint transfinite induction) below. The usage of mix here obeys our condition: its right premise is introduced by $(\yields {!})$, since we are in the left-principal case.

\subsubsection*{Mix Elimination}

We eliminate mix only in the situation where
%
%
%
%
$\Pi \yields {!}^c A$ is obtained by $(\yields {!})$.
We proceed by induction on the rank of the right premise of mix and consider several cases. When using the induction hypothesis for mix, we shall maintain the property that its left premise is principal. In fact, we shall use mix only with the same left premise $\Pi \yields {!}^c A$.

{\em Case 1 (right axiom).} Mix (which is actually a cut) disappears.

{\em Case 2 (right non-principal).} The right premise,
$\Delta_0, {!}^c A, \Delta_1, {!}^c A, \ldots, \Delta_i, {!}^c A, \Delta_{i+1}, {!}^c A,\linebreak \ldots, {!}^c A, \Delta_n \yields B$, is obtained by a rule $\mathrm{R}$ which does not change any of ${!}^c A$. If $\mathrm{R}$ is ``non-branching,'' that is, $(\mconj\yields)$, $(\mathbf{1}\yields)$, $(\adisj\yields)$, $(\aconj\yields)$, $(\KStar\yields)_\omega$, $(\yields\SL)$, $(\yields\BS)$, $(\yields\aconj)$, $(\yields \adisj)$, $({!}\yields)$, $(\PERM)$, $(\NCONTR)$, or $(\WEAK)$, then mix is propagated through $\mathrm{R}$ exactly as cut does.

For ``branching'' rules, $(\SL\yields)$, $(\BS\yields)$, $(\yields\mconj)$, $(\yields\KStar)_n$, the situation is a bit trickier. The instances of ${!}^c A$ can go to different branches, and in this situation we have to apply mix in all such branches and then contract the auxiliary $\Pi$'s into one. We illustrate this on the example of $(\yields\KStar)_3$ with 5 instances of ${!}^c A$ in mix, where two of these instances go to one branch and three to another one (the third branch does not take any). In this situation, mix is applied as follows
$$\small
\infer[(\mathrm{mix})]
{\Delta_0, \Delta_1, \Delta'_2, \Delta''_2, \Delta'''_2, \Delta_3, \Pi, \Delta_4, \Delta_5 \yields B^*}
{\Pi \yields {!}^c A & \infer[(\yields\KStar)]{\Delta_0, {!}^c A, \Delta_1, {!}^c A, \Delta'_2, \Delta''_2, \Delta'''_2, {!}^c A, \Delta_3, {!}^c A, \Delta_4, {!}^c A, \Delta_5 \yields B^*}
{\Delta_0, {!}^c A, \Delta_1, {!}^c A, \Delta'_2 \yields B &
\Delta''_2 \yields B &
\Delta'''_2, {!}^c A, \Delta_3, {!}^c A, \Delta_4, {!}^c A, \Delta_5 \yields B}}
$$
and using mix with smaller ranks of the right premise (while keeping the same left one) we can produce the following derivation
$$
\small
\infer[(\yields\KStar)]
{\Delta_0, \Delta_1, \Pi, \Delta'_2, \Delta''_2, \Delta'''_2, \Delta_3, \Pi, \Delta_4, \Delta_5 \yields B^*}
{\infer[(\mathrm{mix})]{\Delta_0, \Delta_1, \Pi, \Delta'_2 \yields B}{\Pi \yields {!}^c A & \Delta_0, {!}^c A, \Delta_1, {!}^c A, \Delta'_2 \yields B} &
\Delta''_2 \yields B &
\infer[(\mathrm{mix})]{
\Delta'''_2, \Delta_3, \Pi, \Delta_4, \Delta_5 \yields B}
{\Pi \yields {!}^c A & \Delta'''_2, {!}^c A, \Delta_3, {!}^c A, \Delta_4, {!}^c A, \Delta_5 \yields B}}
$$
Now we recall that each formula of $\Pi$ is of the form ${!}^{s_i} C_i$, where $s_i \succcurlyeq c$, thus $s_i \in \Cc$. Therefore, contraction can be applied to $\Pi$ as a whole,
yielding the needed sequent
$$\Delta_0, \Delta_1, \Delta'_2,  \Delta''_2, \Delta'''_2, \Delta_3, \Pi, \Delta_4, \Delta_5 \yields B^*.$$

Finally, propagation of mix through $(\yields {!})$ is performed exactly as for cut.

{\em Case 3 (principal vs.\ principal).} The right premise of mix is obtained by $({!}\yields)$ introducing one of ${!}^c A$. Here mix gets replaced by another mix of a smaller rank and a cut with a formula of smaller complexity ($A$ instead of ${!}^c A$). We illustrate this by an example with $n=3$, where $\Pi$ is put in the place of the first ${!}^c A$ and $({!}\yields)$ introduces the second one:
$$
\infer[(\mathrm{mix})]{\Delta_0, \Pi, \Delta_1, \Delta_2, \Delta_3 \yields B}
{\infer[(\yields {!})]{\Pi \yields {!}^c A}{\Pi \yields A} & \infer[({!}\yields)]
{\Delta_0, {!}^c A, \Delta_1, {!}^c A, \Delta_2, {!}^c A, \Delta_3 \yields B}
{\Delta_0, {!}^c A, \Delta_1, A, \Delta_2, {!}^c A, \Delta_3 \yields B}}
$$
transforms into
$$
\infer[(\NCONTR)\mbox{ several times}]
{\Delta_0, \Pi, \Delta_1, \Delta_2, \Delta_3 \yields B}
{\infer[(\CUT)]
{\Delta_0, \Pi, \Delta_1, \Pi, \Delta_2, \Delta_3 \yields B}
{\Pi \yields A & \infer[(\mathrm{mix})]{\Delta_0, \Pi, \Delta_1, A, \Delta_2, \Delta_3 \yields B}{\Pi \yields {!}^c A &
\Delta_0, {!}^c A, \Delta_1, A, \Delta_2, {!}^c A, \Delta_3 \yields B}}}
$$
Here $(\NCONTR)$ gets applied to $\Pi$ as a whole, since it consists of formulae with subexponentials with labels $r_i \succcurlyeq c$.

Notice that this also works when $\Pi$ is put in place of ${!}^c A$ introduced by $({!}\yields)$, provided there is another instance of ${!}^c A$ in mix. If there is no such instance, then mix is actually cut, and this case is considered in cut elimination.

{\em Case 4 (principal vs.\ !-structural).} The right premise is obtained by a structural rule operating one of the instances of ${!}^c A$. If this structural rule is $(\NCONTR)$, it just gets merged with mix, reducing its rank. If it is $(\WEAK)$ or $(\PERM)$, we first consider the case where it does not operate the instance of ${!}^c A$ replaced by $\Pi$ (``active instance''). In this case the structural rule again gets merged with mix, reducing its rank. If $(\PERM)$ is applied to the active instance, then $c \in \Ec$, and so are all $r_i$ in $\Pi$. Thus, $(\PERM)$ can be applied to $\Pi$ as a whole, and mix gets propagated through $(\PERM)$. Finally, if the active instance is weakened, then $c \in \Wc$, and by $\Wc \cap \Cc \subseteq \Ec$ we can apply $(\PERM)$ to $\Pi$. If there is another instance of ${!}^c A$, we proceed as follows:
$$
\infer[(\mathrm{mix})]{\Delta_0, \Pi, \Delta_1, \Delta_2, \Delta_3 \yields B}
{\Pi \yields {!}^c A & \infer[(\WEAK)]{\Delta_0, {!}^c A, \Delta_1, {!}^c A, \Delta_2, {!}^c A, \Delta_3 \yields B}
{\Delta_0, \Delta_1, {!}^c A, \Delta_2, {!}^c A, \Delta_3, \yields B}}
$$
transforms into
$$
\infer[(\PERM)\mbox{ several times}]
{\Delta_0, \Pi, \Delta_1, \Delta_2, \Delta_3 \yields B}
{\infer[(\mathrm{mix})]{\Delta_0, \Delta_1, \Pi, \Delta_2, \Delta_3 \yields B}{\Pi \yields {!}^c A &
\Delta_0, \Delta_1, {!}^c A, \Delta_2, {!}^c A, \Delta_3, \yields B}}
$$
If the active instance is the only one, we are in the cut elimination case.
\end{proof}

Now we are ready to prove the cut elimination theorem in its full generality:

\begin{Th}\label{Th:cutelim}
The cut rule in ${!}\mathrm{ACT}_\omega$ is eliminable, that is, any sequent that can be proved using cut, can be also proved without cut.
\end{Th}

Notice that Theorem~\ref{Th:cutelim} is not a trivial corollary of Theorem~\ref{Th:onecut}, because an infinite derivation could include infinitely many cuts. Proving Theorem~\ref{Th:cutelim} requires yet another transfinite induction.

\begin{proof}
For convenience, we shall write ${\mathcal{S}}_\alpha$ instead of $\Dcut^\alpha (\varnothing)$ (now cut is allowed in derivations).

 Suppose $\alpha$ is the smallest ordinal such that ${\mathcal{S}}_\alpha$ contains a sequent $\Pi \yields B$ which is not provable without cut. Obviously, $\alpha$ should be a successor ordinal, $\alpha = \beta+1$. Therefore, this sequent is obtained by applying a rule to sequents from ${\mathcal{S}}_\beta$. Since $\beta < \alpha$, these sequents can be derived without cut. The rule yielding $\Pi \yields B$ should be cut, otherwise $\Pi \yields B$ is cut-free derivable. However, if the rule is cut, $\Pi \yields B$ is also cut-free derivable by Theorem~\ref{Th:onecut}. Contradiction.
\end{proof}

\section{Complexity aspects}\label{sec-comp}

Recall that $\Dcut$ is the immediate derivability operator of $!\mathrm{ACT}_\omega$, and its least fixed point is exactly the set of theorems provable in this logic.

We start with the interesting case of $\Cc \ne \varnothing$ (that is, at least one subexponential allows non-local contraction) and prove, under this condition, the following:

\begin{enumerate}
 \item the least fixed point of $\Dcut$ is $\Pi^1_1$-complete (that is, the derivability problem for $!\mathrm{ACT}_\omega$ is $\Pi_1^1$-complete);

\item the closure ordinal of $\Dcut$ is $\omega_1^{\mathrm{CK}}$.
\end{enumerate}

Next, we also consider the case of $\Cc = \varnothing$ (no subexponential allows contraction) and show that complexity there does not raise if compared to $\mathrm{ACT}_\omega$ without subexponentials. Namely, we prove that the closure ordinal is bounded by $\omega^\omega$ and that the derivability problem belongs to $\Pi_1^0$ (and, by Buszkowski's lower bound~\cite{Buszkowski-2007}, it is $\Pi_1^0$-complete).

It is interesting to compare these results with previously known ones for fragments of $!\mathrm{ACT}_\omega$. Recall that $!\mathrm{ACT}_\omega$ is a combination of two systems, infinitary action logic $\mathrm{ACT}_\omega$ and multiplicative-additive Lambek calculus with subexponentials, denoted by $\mathrm{SMALC}_\Sigma$ (here $\Sigma$ is the subexponential signature). These two systems are both extensions of the multiplicative-additive Lambek calculus, $\mathrm{MALC}$, which is, in its turn, an extension of the purely multiplicative Lambek calculus, $\mathrm{L}$.

The following table summarizes the complexity results for these systems:
\begin{center}
 \begin{tabular}{|l|l|l|}\hline
  \multicolumn{1}{|c|}{\bf System} &
  \multicolumn{1}{|c|}{\bf Complexity} &
  \multicolumn{1}{|c|}{\bf Reference(s)} \\\hline
  \strut $\mathrm{L}$ & NP-complete & Pentus, 2006~\cite{Pentus-2006} \\\hline
  \strut $\mathrm{MALC}$ & PSPACE-complete & Kanovich, 1994~\cite{Kanovich-1994}; \\
   & & Kanovich et al., 2019~\cite{kkns-wollic-2019} \\\hline
  \strut $\mathrm{ACT}_\omega$ & $\Pi_1^0$-complete &
  Buszkowski, 2007~\cite{Buszkowski-2007}; \\
  & & Palka, 2007~\cite{Palka-2007} \\\hline
  \strut $\mathrm{SMALC}_\Sigma$ with $\Cc = \varnothing$ &
  PSPACE-complete & \multirow{2}{*}{Kanovich et al., 2018~\cite{KKNS-mscs-2018}} \\\cline{1-2}
  \strut $\mathrm{SMALC}_\Sigma$ with $\Cc \ne \varnothing$ &
  $\Sigma_1^0$-complete & 
  \\\hline
  \strut $!\mathrm{ACT}_\omega$ with $\Cc = \varnothing$ & $\Pi^0_1$-complete & \multirow{2}{*}{this article} \\\cline{1-2}
  \strut $!\mathrm{ACT}_\omega$ with $\Cc \ne \varnothing$ & $\Pi_1^1$-complete &
  \\\hline
 \end{tabular}

\end{center}

From this table, we see that the two sources of undecidability are the Kleene star and the subexponential which allows non-local contraction ($!^c$ with $c \in \Cc$). Another observation is that only the combination of these two yields a system which is not hyperarithmetical.

\subsection{$\mathbf{\Pi}_1^1$-boundedness of $!\mathbf{ACT}_\omega$}

Let us start by establishing that the derivability operator of $! \mathrm{ACT}_{\omega}$ is arithmetical, and can be presented in a `positive' form.
To be precise, call an $\mathcal{L}_2$-formula $\Phi \left( \dots, X, \dots \right)$ \emph{positive in $X$} iff no free occurrence of $X$ in $\Phi$
is in the scope of an odd number of nested negations.\footnote{Remember,
$\rightarrow$ is treated as defined, not as primitive.}
Then:

\begin{Prop} \label{prop-id-d}
There exists an arithmetical formula $\Phi \left( x, X \right)$ positive in $X$ such that for all $S \subseteq \mathrm{Seq}$,
\[
{{\Dcut} \left( S \right)}\ =\
{\left\{ s \in \mathrm{Seq} \mid \mathfrak{N} \vDash \Phi \left( {\sharp s}, {\sharp S} \right) \right\}}
\]
where $\sharp s$ and $\sharp S$ denote the G\"{o}del number of $s$ and the set of G\"{o}del numbers of elements of $S$ respectively.
\end{Prop}

\begin{proof}
Since $! \mathrm{ACT}_{\omega}$ consists of finitely many rules, and the arithmetical formulas positive in $X$ are closed under finite disjunction, we
only need to show that for each rule $\mathrm{R}$ of $! \mathrm{ACT}_{\omega}$ there is an arithmetical formula $\Phi_{\mathrm{R}} \left( x, X \right)$
positive in $X$ such that for any $s \in \mathrm{Seq}$ and $S \subseteq \mathrm{Seq}$,
\[
{\mathfrak{N} \vDash \Phi_{\mathrm{R}} \left( {\sharp s}, {\sharp S} \right)}
\quad \Longleftrightarrow \quad
\begin{array}{c}
s ~ \text{can be obtained from}\vspace{0.5mm}\\
\text{elements of} ~ S ~ \text{by one application of} ~ \mathrm{R} .
\end{array}
\]
(Remember, every scheme of $! \mathrm{ACT}_{\omega}$, e.g.\ `${\left( \vdash {}^{\ast} \right)}_{n}, ~ {n \geqslant 0}$', is viewed as a single rule.)

Probably the most interesting case is where $\mathrm{R} = {\left( ^{\ast} \vdash \right)}_{\omega}$~--- because it deals with $\omega$ premises. To this
end, take
\begin{multline*}
\qquad \quad
P\ :=\
{\{
\left( n, {\sharp \left( \Gamma, A^{\ast}, \Delta \vdash C \right)}, {\sharp \left( \Gamma, A^{n}, \Delta \vdash C \right)} \right) \mid}\\
{n \in \mathbb{N}, ~ \Gamma \in \mathrm{List}, ~ A \in \mathrm{Form}, ~ \Delta \in \mathrm{List} ~ \text{and} ~ C \in \mathrm{Form}
\}} .
\qquad \quad
\end{multline*}
Evidently $P$, being a computable set, is definable in $\mathfrak{N}$ by some arithmetical formula $\Psi \left( x, y, z \right)$. So let
\[
{\Phi_{\mathrm{R}} \left( x, X \right)}\ :=\
{\forall y}\, {\exists z}\, {\left( \Psi \left( y, x, z \right) \wedge z \in X \right)} .
\]
One readily checks that $\Phi_{\mathrm{R}}$ does the job.

Similar but easier arguments cover the other cases.
\end{proof}

Using this fact we can get:

\begin{Th} \label{th-com-d-1}
The derivability problem for
$!\mathrm{ACT}_\omega$ is $\Pi^1_1$-bounded.
\end{Th}

\begin{proof}
%
By Proposition~\ref{prop-id-d}, $\Dcut$ is an arithmetical operator, hence also a $\Pi^1_1$-operator. Therefore its least fixed
point~--- which coincides with the set of sequents provable in $!\mathrm{ACT}_\omega$~--- must be $\Pi^1_1$-bounded by Folklore~\ref{folk-id-2}.
%
\end{proof}

In the next subsection we prove that $!\mathrm{ACT}_\omega$ is $\Pi_1^1$-hard, provided that $\Cc \ne \varnothing$.


\subsection{$\mathbf{\Pi}^1_1$-hardness of $!\mathbf{ACT}_\omega$}

In order to prove $\Pi_1^1$-hardness of $!\mathrm{ACT}_{\omega}$, we wish to use Kozen's result~\cite{Kozen-2002}, which establishes the same complexity bound for deciding entailment of an equation from a finite set of equations in *-continuous Kleene algebras. Since Kozen's result is formulated in the restricted language of Kleene algebras, that is, in the language of only $\mconj$, $\Z$, $\mathbf{1}$, $\adisj$, and $\KStar$, we first formulate the correspondent fragment of our logic. By $\mathrm{KA}_\omega$ we denote the logic of *-continuous Kleene algebras, defined by taking axioms $(\mathrm{id})$, $(\yields\mathbf{1})$, $(\Z\yields)$ and the following rules of $\mathrm{ACT}_\omega$:
$(\mconj\yields)$, $(\yields\mconj)$, $(\mathbf{1}\yields)$,
$(\adisj\yields)$,
$(\yields\adisj)$, $(\KStar\yields)_\omega$, $(\yields\KStar)_n$,
and $(\mathrm{cut})$. If $\Ec = \{ U_1 \yields V_1, \ldots,
U_n \yields V_n \}$, then  $\mathrm{KA}_\omega + \Ec$ denotes
$\mathrm{KA}_\omega$ extended with $\Ec$ as a set of additional axioms.
For sequents derivable in $\mathrm{KA}_\omega + \Ec$ we say that they are ``derivable in $\mathrm{KA}_\omega$ from $\Ec$.''

We show that ${!}$ and $\SL$ allow a variant of deduction theorem, internalizing derivability in $\mathrm{KA}_\omega$ from $\Ec$ into ``pure'' derivability in $!\mathrm{ACT}_\omega$, without extra axioms.
The technique used here is goes back to~\cite{KKNS-mscs-2018}; however, in the presence of Kleene star we can achieve higher complexity boundaries.

\begin{Lemm}\label{Lm:roundrobin}
Let $\Pi \yields B$ be a sequent in the language of Kleene algebras
($\mconj$, $\adisj$, $\KStar$) and let  $U_i$ and $V_i$, $i = 1, \ldots, n$, be formulae in the same language.
 Then the following are equivalent:
 \begin{enumerate}
  \item $\Pi \yields B$ is derivable in $\mathrm{KA}_\omega$ from $\Ec = \{ U_1 \yields V_1, \ldots, U_n \yields V_n \}$;
  \item the sequent
  $${!}^c (\mathbf{1} \SL {!}^c (V_1 \SL U_1)),
  {!}^c (V_1 \SL U_1), \ldots, {!}^c (\mathbf{1} \SL (V_n \SL U_n)),
  {!}^c (V_n \SL U_n), \Pi \yields B$$
  is derivable in $!\mathrm{ACT}_\omega$, where $c \in \Cc$;
  \item the sequent
  $$
  {!}^s (V_1 \SL U_1), \ldots, {!}^s (V_n \SL U_n) \yields B
  $$
  is derivable in $!\mathrm{ACT}_\omega$, where $s \in \Wc \cap \Cc$.
 \end{enumerate}
\end{Lemm}

\begin{proof}

\fbox{$1 \Rightarrow 2$}
Consider a derivation of $\Pi \yields B$ in $\mathrm{KA}_\omega + \mathcal{E}$ and transform it into a derivation of
${!}^c (\mathbf{1} \SL {!}^c A_1), {!}^c A_1, \ldots, {!}^c (\mathbf{1} \SL {!}^c A_n), {!}^c A_n, \Pi \yields B$
in $! \mathrm{ACT}_\omega$.

For brevity, denote ${!}^c (\mathbf{1} \SL {!}^c A_1), {!}^c A_1, \ldots, {!}^c (\mathbf{1} \SL {!}^c A_n), {!}^c A_n$ by ${!}^c\Upsilon$. First let us show that ${!}^c \Upsilon$ actually allows weakening (although, in general, ${!}^c$ does not): if $\Pi \yields B$ is derivable, then so is ${!}^c \Upsilon, \Pi \yields B$. For weakening, we use the $\mathbf{1} \SL {!}^c (V_i \SL U_i)$ formulae, which ``cancel'' ${!}^c (V_i \SL U_i)$:
$$
\infer[(!\yields)]{{!}^c (\mathbf{1} \SL {!}^c (V_i \SL U_i)), {!}^c
(V_i \SL U_i), \Pi \yields B}
{\infer[(\SL\yields)]{\mathbf{1} \SL {!}^c (V_i \SL U_i), {!}^c
(V_i \SL U_i), \Pi \yields B}
{{!}^c (V_i \SL U_i) \yields {!}^c (V_i \SL U_i) &
\infer[(\mathbf{1}\yields)]{\mathbf{1}, \Pi \yields B}{\Pi \yields B}}}
$$

This allows, for each axiom $(\mathrm{id})$ of the form $B \yields B$, to derive the corresponding sequent ${!}^c \Upsilon, B \yields B$. Moreover, the same works for $(\yields\mathbf{1})$ and $(\yields\KStar)_0$: $\yields A^*$ and $\yields \mathbf{1}$ transform, respectively, to
${!}^c \Upsilon \yields \mathbf{1}$ and
${!}^c \Upsilon \yields A^*$. The rule $(\yields\KStar)_n$ for $n>0$ is considered below.

As for the new axioms from $\Ec$, the following derivation reduces them to the $(\mathrm{id})$ case (recall that ${!}^c (V_i \SL U_i)$ is a member of ${!}^c \Upsilon$):
$$
\infer[(\NCONTR)]{{!}^c \Upsilon, U_i \yields V_i}
{\infer[(!\yields)]{{!}^c\Upsilon, {!}^c (V_i \SL U_i), U_i \yields V_i}
{\infer[(\SL\yields)]{{!}^c \Upsilon, V_i \SL U_i, U_i \yields V_i}
{U_i \yields U_i & {!}^c \Upsilon, V_i \yields V_i}}}
$$

Rules $(\mconj\yields)$, $(\mathbf{1}\yields)$, $(\KStar\yields)_\omega$, $(\adisj\yields)$, and $(\yields\adisj)$ transform directly: one just adds the ${!}^c\Upsilon$ prefix to $\Pi$ both in the premise(s) and the conclusion.
Translations of $(\yields\mconj)$, $(\yields\KStar)_n$, for $n > 0$,  and $(\mathrm{cut})$ involve contraction.  Consequent applications of $(\NCONTR)$  merge several instances of ${!}^c\Upsilon$ into one:
$$\small
\infer=[(\NCONTR)]{{!}^c \Upsilon, \Gamma, \Delta \yields A \mconj B}
{\infer[(\mconj\yields)]{{!}^c\Upsilon, \Gamma, {!}^c \Upsilon, \Delta \yields A \mconj B}{{!}^c \Upsilon, \Gamma \yields A &
{!}^c \Upsilon, \Delta \yields B}}
\qquad\small
\infer=[(\NCONTR)]{{!}^c \Upsilon, \Pi_1, \Pi_2, \ldots, \Pi_n \yields A^*}
{\infer[(\KStar\yields)_n]{{!}^c \Upsilon, \Pi_1, {!}^c \Upsilon, \Pi_2, \ldots, {!}^c \Upsilon, \Pi_n \yields A^*}
{{!}^c \Upsilon, \Pi_1 \yields A & \ldots & {!}^c \Upsilon, \Pi_n \yields A}}
$$
$$\small
\infer=[(\NCONTR)]{{!}^c \Upsilon, \Gamma, \Pi, \Delta \yields B}
{\infer[(\CUT)]{{!}^c \Upsilon, \Gamma, {!}^c \Upsilon, \Pi, \Delta \yields B}{{!}^c \Upsilon, \Pi \yields A & {!}^c \Upsilon, \Gamma, A, \Delta \yields B}}
$$

Notice that translation of $(\yields\BS)$ would have required permutation to move ${!}^c \Upsilon$ to the correct place. Fortunately, we encode only $\mathrm{KA}_\omega$ derivations, which do not involve division operations.

The translation presented above yields a derivation of ${!}^c\Upsilon, \Pi \yields B$ in $!\mathrm{ACT}_\omega$. By Theorem~\ref{Th:cutelim}, this derivation can be made cut-free.

 \fbox{$2 \Rightarrow 3$}
Replace ${!}^c$ with ${!}^s$. The latter obeys all the rules which the former does, so we obtain derivability of
$$
{!}^s (\mathbf{1} \SL {!}^s (V_1 \SL U_1)), {!}^s (V_1 \SL U_1), \ldots, {!}^s (\mathbf{1} \SL {!}^s (V_n \SL U_n)), {!}^s (V_n \SL U_n), \Pi \yields B
$$
Sequents $\yields {!}^s (\mathbf{1} \SL
{!}^s (V_i \SL U_i))$ are derivable using the weakening rule:
$$
\infer[(\yields {!})]{\yields {!}^s (\mathbf{1} \SL {!}^s (V_i \SL U_i))}
{\infer[(\yields\SL)]{\yields \mathbf{1} \SL {!}^s (V_i \SL U_i)}
{\infer[(\mathrm{weak})]{{!}^s (V_i \SL U_i) \yields \mathbf{1}}
{\yields\mathbf{1}}}}
$$
Using cut, we obtain the needed sequent
$$
{!}^s (V_1 \SL U_1), \ldots, {!}^s (V_n \SL U_n), \Pi \yields B.
$$

\fbox{$3 \Rightarrow 1$}
Consider a cut-free proof of ${!}^s (V_1 \SL U_1), \ldots, {!}^s (V_n \SL U_n), \Pi \yields B$ and erase all formulae including $\SL$ in it. (Recall that $U_i$, $V_j$, $\Pi$, and $B$ include only $\mconj$, $\adisj$, and $\KStar$.)
In particular, all ${!}^s$-formulae get erased, and the goal sequent becomes the original $\Pi \yields B$.

After erasing, all rules operating ${!}^s$ trivialize, and applications of $(\SL\yields)$ transform into
$$
\infer{\Gamma, \Pi, \Delta \yields B}
{\Pi \yields U_i & \Gamma, V_i, \Delta \yields B}
$$
Here $V_i \SL U_i$ to the left of $\Pi$ gets hidden. This gets modelled by two cuts:
$$
\infer[(\CUT)]{\Gamma, \Pi, \Delta \yields B}
{\Pi \yields U_i & \infer[(\CUT)]{\Gamma, U_i, \Delta \yields B}
{U_i \yields V_i & \Gamma, V_i, \Delta \yields B}}
$$
Thus, we obtain a derivation of $\Pi \yields B$ in $\mathrm{KA}_\omega + \Ec$.
 \end{proof}

 Now we are almost ready to prove $\Pi_1^1$-hardness of $!\mathrm{ACT}_\omega$ by reduction from derivability in $\mathrm{KA}_\omega$ from finite sets $\Ec$.
 Kozen's result, however, is formulated semantically: it establishes complexity of the universal Horn theory of *-continuous Kleene algebras, that is, the problem of whether a given sequent $\Pi \yields B$ is {\em true} under all interpretations in all *-continuous Kleene algebras in which all sequents in $\Ec$ are true. In order to shift to syntax, namely, derivability in $\mathrm{KA}_\omega$ from $\Ec$, one needs a completeness theorem.

 Let us give the formal definitions and statements.

 \begin{Def}
  A \emph{*-continuous Kleene algebra} is an partially ordered algebraic structure \linebreak $(\mathfrak{A}, {\preccurlyeq}, \mconj, \mathbf{1}, \adisj, \KStar)$ where $(\mathfrak{A}, \adisj, \mconj, \Z, \U)$ is an idempotent semi-ring (idempotency means that $a \adisj a = a$), the partial order is defined as follows: $a \preccurlyeq b$ if and only if $b = a \adisj b$; and
  %
  $a \mconj b^* \mconj c = \sup\{a \mconj b^n \mconj c \mid n \geqslant 0\}$ (for any
  $a,b,c\in\mathfrak{A}$). (Here the supremum is taken w.r.t.\ the $\preccurlyeq$ partial order.)
 \end{Def}

 One can easily see that *-continuous Kleene algebras are exactly the algebraic structures which satisfy the $\mathrm{KA}_\omega$ theory (where $\yields$ stands for $\preccurlyeq$ and commas in left-hand sides of sequents are interpreted as $\mconj$). In particular, the *-continuity condition for Kleene star, $a \mconj b^* \mconj c = \sup\{a \mconj b^n \mconj c \mid n \geqslant 0\}$, corresponds exactly to the $\omega$-rule $(\KStar\yields)_\omega$.

 \begin{Def}
  An \emph{interpretation} of $\mathrm{KA}_\omega$ formulae in a *-continuous Kleene algebra $\mathfrak{A}$ is a function $\alpha$ mapping formulae to elements of $\mathfrak{A}$, which is defined in an arbitrary way on variables and commutes with operations. A sequent  $A_1, \ldots, A_n \yields B$ is \emph{true} under interpretation $\alpha$ if $\alpha(A_1)\cdot\ldots\cdot \alpha(A_n) \preccurlyeq \alpha(B)$; in the special case of $n = 0$, truth of $\yields B$ under $\alpha$ means $\mathbf{1} \preccurlyeq \alpha(B)$.
 \end{Def}

 \begin{Def}
  A sequent $\Pi \yields A$ is \emph{entailed} by $\Ec$ (on *-continuous Kleene algebras) iff it is true under all interpretations on *-continuous Kleene algebras, under which all sequents from $\Ec$ are true.
 \end{Def}

Completeness theorem for $\mathrm{KA}_\omega + \Ec$ for an arbitrary $\Ec$, that is, {\em strong completeness} of $\mathrm{KA}_\omega$, is formulated as follows:

\begin{Th}\label{Th:LT}
 A sequent is derivable in $\mathrm{KA}_\omega + \Ec$ if and only if it is entailed by $\Ec$ on *-continuous Kleene algebras.
\end{Th}

\begin{proof}
 The ``only if'' part (soundness) is established by a routine check that axioms of $\mathrm{KA}_\omega$ are true under all intepretations and that rules of $\mathrm{KA}_\omega$ are truth-preserving.

 For the ``if'' part (completeness), we apply standard Lindenbaum -- Tarski construction, relativized to $\Ec$. (Kozen uses a factor-algebra of the algebra of regular expressions, $\mathrm{REG}(\Sigma^*) / \Ec$, instead~\cite[Lemma~4.1]{Kozen-2002}.)

 Let $\mathcal{F}$ denote the set of all formulae in the language of $\mathrm{KA}_\omega$. Two formulae, $A$ and $A'$, are equivalent in $\mathrm{KA}_\omega + \Ec$, if both $A \yields A'$ and $A' \yields A$ are derivable in $\mathrm{KA}_\omega$ from $\Ec$. Due to the $(\mathrm{id})$ axiom and the cut rule, this is indeed an equivalence relation. Denote the factor-set (set of equivalence classes) by $\mathcal{F} / \Ec$. The equivalence class of $A$ is denoted by $[A]_\Ec$.

 Next, impose a structure of *-continuous Kleene algebra on $\mathcal{F} / \Ec$:
 \begin{align*}
  & [ A ]_\Ec \mconj [ B ]_\Ec = [ A \mconj B ]_\Ec \\
  & [ A ]_\Ec \adisj [ B ]_\Ec = [ A \adisj B ]_\Ec \\
  & [ A ]_\Ec^* = [ A^* ]_\Ec
 \end{align*}
 The unit is $[\mathbf{1}]_\Ec$. This definition is correct, because our equivalence relation is a congruence w.r.t. Kleene algebra operations: if $A$ is equivalent to $A'$ and $B$ is equivalent to $B'$, then
 so are $A \mconj B$ and $A' \mconj B'$, $A \adisj B$ and $A' \mconj B'$, and $A^*$ and ${A'}^*$ (the last equivalence involves the $\omega$-rule to establish).

 A routine check shows that $\mathcal{F} / \Ec$ is indeed a *-continuous Kleene algebra. The standard interpretation is defined as follows:
 $\alpha(A) = [A]_\Ec$ (by definition, it commutes with operations). All sequents from $\Ec$ are true under this interpretation. Indeed, if $(U \yields V) \in \Ec$, then $V$ is equivalent to $U \adisj V$, whence
 $\alpha(V) = \alpha(U) \adisj \alpha(V)$, that is, $\alpha(U) \preccurlyeq \alpha(V)$.

 Since rules of $\mathrm{KA}_\omega$ are truth-preserving, every sequent derivable from $\Ec$ is also true under interpretation $\alpha$. Moreover, the converse also holds. If $A \yields B$ is true under $\alpha$, then $\alpha(A \adisj B) = \alpha(B)$. Therefore, $A \adisj B$ is equivalent to $B$, in particular, $A \adisj B \yields B$ is derivable from $\Ec$. By cut with $A \yields A \adisj B$ we get derivability of $A \yields B$. For sequents with zero ($\yields B$) or more than one ($A_1, \ldots, A_n \yields B$) formulae in the left-hand side, take $A = \mathbf{1}$ or $A = A_1 \mconj\ldots\mconj A_n$ respectively.

 Thus, $\mathcal{F} / \Ec$ gives a {\em universal model} for $\mathrm{KA}_\omega + \Ec$. Any formula which is true under all interpretations, under which $\Ec$ is true, is in particular true under $\alpha$ in $\mathcal{F} / \Ec$ and therefore derivable from $\Ec$. This finishes the completeness proof.
\end{proof}

Let us recall Kozen's theorem in the formulation we are going to use:
\begin{Th}
 The following problem is $\Pi_1^1$-complete: given a finite $\Ec = \{ U_1 \yields V_1, \ldots, U_n \yields V_n \}$ and $\Pi \yields B$ in the language of $\mathrm{KA}_\omega$, determine whether $\Pi \yields B$ is entailed by $\Ec$ on *-continuous Kleene algebras.
\end{Th}

This theorem indeed follows from the reasoning of Kozen's article~\cite{Kozen-2002}, via the following technical remark. Kozen's original formulation uses equations instead of inequations, in other words, $=$ instead of $\yields$. In $\Ec$, this is does not make any change, since $U_i \yields V_i$ can be equivalently replaced with $V_i = U_i \adisj V_i$ (by definition) and, vice versa, instead of $U_i = V_i$ one may consider two hypotheses, $U_i \yields V_i$ and $V_i \yields U_i$. For the conclusion, $\Pi \yields B$, this could be a more important issue, since an equation here would split into {\em two} inequations. Fortunately, in Kozen's construction~\cite[Lemma~5.1]{Kozen-2002} the conclusion is actually an inequation (that is, an equation of the form $V = U \adisj V$), so this issue disappears, and we can state $\Pi_1^1$-hardness in the inequational language as well.

 Using this theorem, Lemma~\ref{Lm:roundrobin}, and Theorem~\ref{Th:LT}, we can establish $\Pi_1^1$-hardness of the derivability problem for ${!}\mathrm{ACT}_\omega$.

\begin{Th} \label{Cne0-hard}
Suppose $\Cc \ne \varnothing$. Then the derivability problem for ${!}\mathrm{ACT}_\omega$ is $\Pi_1^1$-hard.
\end{Th}

 \begin{proof}
 Let the subexponential signature $\Sigma$ include a label $c$ such that $c \in \Cc$.

  We proceed by the following reduction: for a pair $\langle \Pi \yields B, \Ec \rangle$ of a sequent and a finite set of sequents in the language of Kleene algebra, where $\Ec = \{ U_1 \yields V_1, \ldots, U_n \yields V_n \}$, let
  \begin{multline*}
  f(\langle \Pi \yields B, \Ec \rangle) =\\
  {!}^c (\mathbf{1} \SL {!}^c (V_1 \SL U_1)),
  {!}^c (V_1 \SL U_1), \ldots,
  {!}^c (\mathbf{1} \SL {!}^c (V_n \SL U_n)),
  {!}^c (V_n \SL U_n), \ldots, \Pi \yields B.
  \end{multline*}

  This function $f$ provides the needed $m$-reduction from Kozen's $\Pi_1^1$-complete entailment problem to the derivability problem for ${!}\mathrm{ACT}_\omega$:
  \[
   \begin{aligned}
    &\mbox{$\Pi \yields B$ is entailed by $\Ec$}\\
    &\mbox{on *-continuous Kleene algebras}
   \end{aligned}
   \quad\iff\quad
    f(\langle \Pi \yields B, \Ec \rangle)\mbox{ is derivable in }{!}\mathrm{ACT}_\omega.
  \]
  This equivalence is proved via $\mathrm{KA}_\omega + \Ec$. Namely, the following three statements are equivalent:
  \begin{enumerate}
   \item $\Pi \yields B$ is entailed by $\Ec$ on *-continuous Kleene algebras;
   \item $\Pi \yields B$ is derivable in $\mathrm{KA}_\omega + \Ec$;
   \item $f(\langle \Pi \yields B, \Ec \rangle)$ is derivable in ${!}\mathrm{ACT}_\omega$, provided $c \in \Cc$.
  \end{enumerate}
Here statements 1 and 2 are equivalent by Theorem~\ref{Th:LT} and statements 2 and 3 are equivalent by Lemma~\ref{Lm:roundrobin}. This finishes the proof of $\Pi_1^1$-hardness (and, by Theorem~\ref{th-com-d-1}, $\Pi_1^1$-completeness) of the derivability problem for ${!}\mathrm{ACT}_\omega$.
 \end{proof}

\subsection{Closure ordinal of $\Dcut$}
Next, we compute the closure ordinal of $\Dcut$, provided that $\Cc \ne \varnothing$.

The following result should help us familiarise ourselves with transfinite sequences arising in the study of monotone $\Pi^1_1$- and
$\Sigma^1_1$-operators.

\begin{Prop} \label{prop-id-rec}
Let $F: \mathcal{P} \left( \mathbb{N} \right) \rightarrow \mathcal{P} \left( \mathbb{N} \right)$ be a monotone $\Pi^1_1$-operator. Then there exists a
total computable function from $\mathbb{N}$ to $\mathbb{N}$ that, given any $n \in \mathcal{O}$, returns the G\"{o}del number of~a~$\Pi^1_1$-for\-mula
defining $F^{\nu_{\mathcal{O}} \left( n \right)} \left( \varnothing \right)$ in $\mathfrak{N}$. Similarly with $\Sigma^1_1$ in place of $\Pi^1_1$.
\end{Prop}

\begin{proof}
Let $F$ be as described. In particular, $F = \left[ \Phi \right]$ for some $\Pi^1_1$-formula $\Phi \left( x, X \right)$. Notice that since $F$ is
monotone, it is also expressible by
\[
{\Phi' \left( x, X \right)}\ :=\
{\forall Y}\, {\left(
{\exists y}\, {\left( y \in X \wedge y \not \in Y \right)} \vee
\Phi \left( x, Y \right)
\right)} ,
\]
viz.\ $F = \left[ \Phi' \right]$. Moreover, $\Phi'$ can easily be reduced to a $\Pi^1_1$-formula positive in $X$. Thus, without loss of generality,
we may assume that $\Phi$ is positive in $X$. Next, we remark that the graph of $\ae$ (our universal partial computable function), being a computably
enumerable set, is definable in $\mathfrak{N}$ by some arithmetical formula $\Theta \left( x, y, z \right)$. For every $n \in \mathbb{N}$, let
\begin{align*}
{\Phi_{n} \left( x \right)}\ &:=\
{\forall Y}\, {\left(
{\exists y}\, {\left( {\Pi^1_1 \text{-} \mathrm{SAT}} \left( y, \underline{n} \right) \wedge y \not \in Y \right)} \vee
\Phi \left( x, Y \right)
\right)} ,\\
{\Phi_{n}^{\star} \left( x \right)}\ &:=\
{\exists y}\, {\exists z}\, {\left(
\Theta \left( \underline{n}, y, z \right) \wedge
{\Pi^1_1 \text{-} \mathrm{SAT}} \left( x, z \right)
\right)}
\end{align*}
with $\Pi^1_1 \text{-} \mathrm{SAT}$ as in Folklore~\ref{folk-msoa-2} and $\underline{n}$ denoting the numeral for $n$. It is not hard to reduce each
of $\Phi_{n} \left( x \right)$ and $\Phi_{n}^{\star} \left( x \right)$ to a $\Pi^1_1$-form. Further, we make the following observations.
\begin{enumerate} \renewcommand{\theenumi}{\roman{enumi}}

\item Suppose $n = \# \Psi \left( x \right)$ with $\Psi \left( x \right)$ a $\Pi^1_1$-formula, and let $P$ be the set defined in $\mathfrak{N}$ by
$\Psi \left( x \right)$. Then $\Phi_{n} \left( x \right)$ defines $F \left[ P \right]$ in $\mathfrak{N}$.

\item Suppose that $\ae_n \left[ \mathbb{N} \right] = \left\{ {\# \Psi_0 \left( x \right)}, {\# \Psi_1 \left( x \right)}, \dots \right\}$ where
$\Psi_0 \left( x \right)$, $\Psi_1 \left( x \right)$, \dots\ are $\Pi^1_1$-for\-mu\-las,  and for any $k \in \mathbb{N}$, let $P_k$ be the set defined
in $\mathfrak{N}$ by $\Psi_k \left( x \right)$. Then $\Phi_{n}^{\star} \left( x \right)$ defines $\bigcup_{k \in \mathbb{N}} P_k$ in $\mathfrak{N}$.

\end{enumerate}
Accordingly we have computable functions $f$ and $f^{\star}$ such that for every $n \in \mathbb{N}$:
\begin{enumerate} \renewcommand{\theenumi}{\Roman{enumi}}

\item if $n = \# \Psi \left( x \right)$ with $\Psi \left( x \right)$ a $\Pi^1_1$-formula, then $f \left( n \right) = \# \Omega \left( x \right)$ with
$\Omega \left( x \right)$ a $\Pi^1_1$-for\-mu\-la that is equivalent to $\Phi_{n} \left( x \right)$;

\item if $\ae_n \left[ \mathbb{N} \right] = \left\{ {\# \Psi_0 \left( x \right)}, {\# \Psi_1 \left( x \right)}, \dots \right\}$ where $\Psi_0
\left( x \right)$, $\Psi_1 \left( x \right)$, \dots\ are $\Pi^1_1$-for\-mu\-las, then $f^{\star} \left( n \right) \linebreak
= \# \Omega \left( x \right)$ where $\Omega \left( x \right)$ is equivalent to $\Phi_{n}^{\star} \left( x \right)$.

\end{enumerate}
Now take a total computable two-place function $g$ such that for any $e, k, n \in \mathbb{N}$,
\[
{\ae_{g \left( e, k \right)} \left( n \right)}\ =\ {\ae_e \left( \ae_k \left( n \right) \right)}
\]
(provided by the $s$-$m$-$n$ theorem), and let $h$ be a total computable func\-tion which satisfies
\[
{\ae_{h \left( e \right)} \left( n \right)}\ =\
\begin{cases}
{f \left( \ae_e \left( k \right) \right)} &\text{if} ~\, n = 2^k \ne 1 ,\\
{f^{\star} \left( g \left( e, k \right) \right)} &\text{if} ~\, n = 3 \times 5^k ,\\
\# {{\forall X}\, {\left( x \ne x \right)}} &\text{otherwise}\\
\end{cases}
\]
for all $e, n \in \mathbb{N}$. By the recursion theorem, $\ae_{h \left( c \right)} = \ae_c$ for some $c \in \mathbb{N}$. As can easily be verified,
the function $\ae_c$ does the job.\footnote{In
particular, $\ae_c$ must be total. For otherwise let $n$ be the least element of $\mathbb{N} \setminus \mathrm{dom} \left( \ae_c \right)$. Then
$n = 2^k$ with $k \ne 0$, and therefore $\ae_c \left( k \right)$ is undefined, which contradicts the choice of $n$.}

A similar argument works for $\Sigma^1_1$.
\end{proof}

This result also provides a useful tool for calculating closure ordinals:

\begin{Prop} \label{prop-co-hyp}
Let $F: \mathcal{P} \left( \mathbb{N} \right) \rightarrow \mathcal{P} \left( \mathbb{N} \right)$ be a monotone hyperarithmetical operator whose least
fixed point is~not hyperarithmetical. Then the closure ordinal of $F$ is $\omega_1^{\mathrm{CK}}$.
\end{Prop}

\begin{proof}
Let $F$ be as described, and take $\alpha$ to be its closure ordinal, which exists by Folklore~\ref{folk-id-1}. Clearly $\alpha \leqslant
\omega_1^{\mathrm{CK}}$ by Folklore~\ref{folk-id-2}. Now suppose that $\alpha < \omega_1^{\mathrm{CK}}$, so $\alpha \in \mathsf{C} \text{-}
\mathsf{Ord}$ (i.e.\ $\alpha = \nu_{\mathcal{O}} \left( n \right)$ for some $n \in \mathcal{O}$). Since $F$ is both a $\Pi^1_1$-operator and
a $\Sigma^1_1$-operator, $F^{\alpha} \left( \varnothing \right)$ must be hy\-per\-a\-rith\-me\-ti\-cal by Proposition~\ref{prop-id-rec}. At
the same time, $F^{\alpha} \left( \varnothing \right)$ is the the least fixed point of $F$~--- and thus we get a contradiction. Consequently
$\alpha = \omega_1^{\mathrm{CK}}$.
\end{proof}

This quickly leads to:

\begin{Th} \label{th-com-d-2}
Suppose $\Cc \ne \varnothing$. Then the closure ordinal of $\Dcut$ is $\omega_1^{\mathrm{CK}}$. The same holds for $\Dnocut$.
\end{Th}

\begin{proof}
Clearly $\Dcut$ is a hyperarithmetical operator by Proposition~\ref{prop-id-d}, and the least fixed point of $\Dcut$
is~not hyperarithmetical by Theorem~\ref{Cne0-hard} and Folklore~\ref{folk-msoa-3}. So the result follows by Proposition~\ref{prop-co-hyp}.
The same argument works for $\Dnocut$.
\end{proof}



\subsection{The fragment without contraction}

The case of $\Cc = \varnothing$, that is, no subexponential allows contraction, is significantly different from $\Cc \ne \varnothing$. Namely, now we have no opportunity to encode entailment from finite sets of sequents, so the only lower bound we still have is $\Pi_1^0$-hardness (which holds already for $\ACTomega$ without subexponentials~\cite{Buszkowski-2007}).  We shall show that this bound is tight, that is, that the derivability problem for $! \ACTomega$ with $\Cc = \varnothing$ is also $\Pi^0_1$-bounded.\footnote{For
more information on $\Sigma^0_1$- and $\Pi^0_1$-sets one may consult \cite{Rogers-1967} and \cite{HajekPudlak-1993}, for example. In particular, it should be remarked that there exists a $\Pi^0_1$-formula $\Pi^0_1 \text{-} \mathrm{SAT} \left( x, y \right)$ such that for any $\Pi^0_1$-formula $\Phi \left( x \right)$,
\[
\mathfrak{N} \vDash
{\forall x}\,
{\left(
\Pi^0_1  \text{-} \mathrm{SAT} \left( x, {\# \Phi} \right) \leftrightarrow \Phi \left( x \right)
\right)} ,
\]
and similarly with $\Sigma^0_1$ in place of $\Pi^0_1$~--- cf.\ Section 1(d) of Chapter~I in \cite{HajekPudlak-1993}.}

In principle, the $\Pi^0_1$-boundedness can be shown by extending the corresponding techniques used for $\ACTomega$: Palka's *-elimination~\cite{Palka-2007} or a calculus with non-well-founded proofs by Das and Pous~\cite{DasPous-2018}. However, we shall develop a new method for proving $\Pi^0_1$-boundedness, which will be more independent from concrete structural properties of the proof system.

We start with proving an upper bound for the closure ordinal of the fragment $\eACTomega$ with $\Cc = \varnothing$.\footnote{Palka~\cite{Palka-2007} claims, without a proof, an $\omega_1$ upper bound for this ordinal for $\ACTomega$, by saying that $\Snocut_{\omega_1}$ is the set of all derivable sequents. We give a better upper bound, which is useful for reasoning about complexity of derivability.} Following~\cite{Palka-2007}, we define a complexity parameter on formulae and sequents. The values of this complexity parameter will be not natural numbers, but rather elements of a countable well-founded ordered set. Let $\Nc$ be the set of all infinite sequences of natural numbers which eventually stabilize at zero:
$$
\Nc = \{ (m_0, m_1, m_2, \ldots, m_n, \ldots) \mid \exists i_0 \, \forall i \geqslant i_0 \ m_i = 0 \}.
$$
On $\Nc$, let us define the anti-lexicographical order and two operations, pointwise sum and lifting:
\begin{align*}
& (m_0,m_1,\ldots) \prec (n_0,n_1,\ldots) \iff \exists j_0 \, (m_{j_0} < n_{j_0} \mbox{ and } \forall j > j_0\ m_j = n_j); \\
& (m_0,m_1,\ldots) \ncplus (n_0,n_1,\ldots) = (m_0 + n_0, m_1 + n_1, \ldots);\\
& (m_0,m_1,\ldots) {\uparrow} = (0,m_0,m_1,\ldots).
\end{align*}
Let us also define $\iota = (1,0,0,\ldots)$ as the `unit' in $\Nc$.

It is easy to see that $(\Nc, \prec)$ is a well-founded linearly ordered set, and its order type is $\omega^\omega$. Indeed, this is established by the following isomorphism $\nu \colon \Nc \to \omega^\omega$:
\[
{\nu (m_0, m_1, m_2, \ldots, m_n, \ldots)}\ :=\
\ldots + \omega^n \cdot m_n + \ldots + \omega^2 \cdot m_2 + \omega \cdot m_1 + m_0 .
\]
(The number of summands on the right-hand side is always finite, since sequences in $\Nc$ stabilize at zero.)

Now let us define the complexity measure $\eta(\cdot)$ on formulae and sequents, with values in $\Nc$:
\begin{align*}
& \eta(p_i) = \iota \mbox{\quad for each variable $p_i$;}\\
& \eta(\U) = \iota;\\
& \eta(A \BS B) = \eta(B \SL A) = \eta(A \mconj B) = \eta(A \aconj B) = \eta(A \adisj B) = \eta(A) \ncplus \eta(B) \ncplus \iota;\\
& \eta({!}^s A) = \eta(A) \ncplus \iota \mbox{\quad for each $s \in \Ic$;}\\
& \eta(A^*) = (\eta(A) {\uparrow}) \ncplus \iota,
\end{align*}
and for a sequent $s = A_1, \ldots, A_n \yields B$ let $\eta(s) = \eta(A_1) \ncplus \ldots \ncplus \eta(A_n) \ncplus \eta(B)$.

In the absence of contraction ($\Cc = \varnothing$), each of the rules of $\eACTomega$ enjoys the following property: if $s$ is the conclusion and $s'$ is one of the premises, then $\eta(s') \preccurlyeq \eta(s)$. Moreover, these inequalities are {\em strict} for all rules, except permutations $(\PERM)_1$ and $(\PERM)_2$, which do not change complexity.

In order to overcome the issue with permutations, let us consider {\em generalized rules.} A generalized rule application consists of a application of a rule which is not a permutation, with an arbitrary number of $(\PERM)_1$ and $(\PERM)_2$ applications below. Replacing rules with their generalized versions and removing permutation rules yields a system equivalent to $\eACTomega$.


Generalized rules keep the good properties of the original rules: correctness of their applications is decidable and, for a given sequent $s$, there is a finite choice of generalized rules applications which could be the lowermost (immediate) ones in a derivation of $s$.

For each generalized rule, if $s$ is its conclusion and $s'$ is one of its premises, we have $\eta(s') \prec \eta(s)$.

Now we are ready to prove an upper bound on the closure ordinal for $\eACTomega$ with $\Cc = \varnothing$. Recall that $\Dcut$ and $\Dnocut$ denote the immediate derivability operator for $\eACTomega$, with and without the cut rule respectively.

\begin{Th}
Suppose $\Cc = \varnothing$. Then the closure ordinal of $\Dnocut$ is less than or equal to $\omega^\omega$; and the same holds for $\Dcut$.
\end{Th}

\begin{proof}
We shall prove the following statement: if sequent $s$ is derivable, then it belongs to $\Snocut_\alpha = \Dnocut^\alpha(\varnothing)$, where $\alpha =  \omega \cdot \nu(\eta(s))$. Notice that here the $\Dnocut$ operator is defined w.r.t.\ the original formulation of the calculus, not the one with generalized rules. Cut is also disallowed: by Theorem~\ref{Th:cutelim}, any derivable sequent is cut-free derivable.

This statement is proved by transfinite induction on $\nu(\eta(s))$. Notice that, by definition of $\eta$, this ordinal is always a successor ($m_0 \ne 0$). Let $\nu(\eta(s)) = \beta + 1$. Consider a derivation of $s$ which uses generalized rules. For each premise $s'$ of the lowermost generalized rule in this derivation, we have $\eta(s') \prec \eta(s)$, thus, $\nu(\eta(s')) \leqslant \beta$. By induction hypothesis, each $s'$ belongs to $\Snocut_{\omega\cdot\nu(\eta(s'))} \subseteq \Snocut_{\omega \cdot \beta}$.

Now let us look inside this generalized rule: it consists of an application of a non-permutation rule, which derives a sequent $\tilde{s}$ from the premises, followed by a finite number $k$ of permutations, which yield the goal sequent $s$. Thus, $\tilde{s} \in \Snocut_{\omega \cdot \beta + 1}$ and $s \in \Snocut_{\omega \cdot \beta + 1+k} \subseteq \Snocut_{\omega \cdot \beta + \omega} = \Snocut_{\omega \cdot (\beta +1)}$, q.e.d.

For any sequent $s$ we have $\nu(\eta(s)) < \omega^\omega$, whence $\alpha = \omega \cdot \nu(\eta(s)) < \omega \cdot \omega^\omega = \omega^\omega$. Therefore, each derivable sequent belongs already to $\Snocut_{\omega^\omega}$, which makes $\omega^\omega$ an upper bound for the closure ordinal.

The result for $\Dcut$ follows easily, since we have $\Scut_\alpha  \supseteq \Snocut_\alpha$ for any $\alpha$.
\end{proof}

Next, to derive the desired complexity results, it is helpful to extend the G\"{o}del numbering $\sharp$ for $\mathrm{Seq}$ in a suitable way.
Let $\mathrm{Seq}^{\ast}$ be the collection of all finite $\mathrm{Seq}$-sequences.
Evidently, we can use $\sharp$ to define an effective numbering $\sharp_{\ast}$ for $\mathrm{Seq}^{\ast}$, so that
\[
E_{\ast}\ :=\
{\left\{
\left( {\sharp s}, {\sharp_{\ast} \left( s_1, \ldots, s_n \right)} \right) \mid
s = s_i ~ \text{for some} ~ i \in \left\{ 1, \ldots, n \right\}
\right\}}
\]
is a computable relation. Further, take $\mathrm{Seq}^{\bullet}$ to be the collection of all infinite $\mathrm{Seq}$-sequences of the form
\[
{\Gamma, A^{0}, \Delta \vdash B} , \quad
{\Gamma, A^{1}, \Delta \vdash B} , \quad
{\Gamma, A^{2}, \Delta \vdash B} , \quad \ldots
\]
--- i.e., all those that may appear above the line in the $\omega$-rule. Since the elements of $\mathrm{Seq}^{\bullet}$ have a  simple form, we shall also assume an effective numbering $\sharp_{\bullet}$ for $\mathrm{Seq}^{\bullet}$, so that
\[
E_{\bullet}\ :=\
{\left\{
\left( {\sharp s}, {\sharp_{\bullet} \left( s_0, s_1, \ldots\; \right)} \right) \mid
s = s_n ~ \text{for some} ~ n \in \mathbb{N}
\right\}}
\]
is computable. Now consider
\begin{multline*}
R_{\ast}\ :=\
{\left\{
\left( {\sharp s}, {\sharp_{\ast} \left( s_1, \dots, s_n \right)} \right) \mid
\begin{tabular}{c}
$s$ can be obtained from $s_1$, \ldots, $s_n$ by one\\
application of a generalized finitary rule
\end{tabular}
\right\}} \quad \text{and}\\[1em]
R_{\bullet}\ :=\
{\left\{
\left( {\sharp s}, {\sharp_{\bullet} \left( s_0, s_1, \ldots\; \right)} \right) \mid
\begin{tabular}{c}
$s$ can be obtained from $s_0$, $s_1$, \ldots\ by one\\
application of the generalized $\omega$-rule
\end{tabular}
\right\}} .
\end{multline*}
Clearly, both $R_{\ast}$ and $R_{\bullet}$ are computable. Also, since axioms are nullary rules, we have
\[
{\left\{ \left( {\sharp s}, {\sharp_{\ast} \left( {~} \right)} \right) \mid s ~ \text{is an axiom} \right\}}\
\subseteq\
R_{\ast}
\]
where $\left( {~} \right)$ denotes the empty sequence of sequents. Finally, to simplify things slightly, we shall occasionally identify elements of $\mathcal{N}$
with those of $\mathbb{N}$.\footnote{Formally,
one might have assumed some effective one-one numbering from $\mathcal{N}$ onto $\mathbb{N}$ (which turns $\prec$ into a computable binary relation).}

The fact that $\eta$ gives us a nice ranking for $\mathrm{Seq}$ in the case when $\mathcal{C}$ is empty can be used to obtain the following result, which is partially similar to Proposition~\ref{prop-id-rec}.

\begin{Prop} \label{Prop-Pi01-levels}
Suppose $\mathcal{C} = \varnothing$. Then there exists a computable function from $\mathcal{N}$ to $\mathbb{N}$ that, given any $n \in \mathcal{N}$, returns
the G\"{o}del num\-ber of a $\Pi^0_1$-formula defining
\[
\left\{ \sharp s \mid \eta \left( s \right) = n ~ \text{and} ~ s ~ \text{is derivable in} ~ !\mathrm{ACT}_{\omega} \right\}
\]
in the standard model $\mathfrak{N}$ of arithmetic.
\end{Prop}

\begin{proof}
Observe that for any sequent $s$,
\[
{\left\{ {\sharp_{\ast} \left( s_1, \ldots, s_n \right)} \mid R_{\ast} \left( {\sharp s} , {\sharp_{\ast} \left( s_1, \ldots, s_n \right)} \right) \right\}}
\quad \text{and} \quad
{\left\{ {\sharp_{\bullet} \left( s_0, s_1, \ldots \right)} \mid R_{\bullet} \left( {\sharp s} , {\sharp_{\bullet} \left( s_0, s_1, \ldots\; \right)} \right) \right\}}
\]
are finite sets, which are computable uniformly in $\sharp s$. Hence we have computable functions $\rho_{\ast}$ and $\rho_{\bullet}$ from $\mathbb{N}$ to $\mathbb{N}$ such that for each sequent $s$,
\begin{align*}
{\rho_{\ast} \left( \sharp s \right)}\ &=\
{\max \left( \left\{ {\sharp_{\ast} \left( s_1, \ldots, s_n \right)} \mid R_{\ast} \left( {\sharp s} , {\sharp_{\ast} \left( s_1, \ldots, s_n \right)} \right) \right\} \cup \left\{ 0 \right\} \right)} ,\\
{\rho_{\bullet} \left( \sharp s \right)}\ &=\
{\max \left( \left\{ {\sharp_{\bullet} \left( s_0, s_1, \ldots\; \right)} \mid R_{\bullet} \left( {\sharp s} , {\sharp_{\bullet} \left( s_0, s_1, \ldots\; \right)} \right) \right\} \cup \left\{ 0 \right\} \right)} .
\end{align*}
Let $\dot{\eta}$ be a computable function from $\mathbb{N}$ to $\mathcal{N}$ such that $\dot{\eta} \left( \sharp s \right) = \eta \left( s \right)$ for all $s \in \mathrm{Seq}$. Now take $\Theta \left( x, y, z \right)$ to be a $\Sigma^0_1$-formula defining the graph of our universal partial computable function $\ae$ in $\mathfrak{N}$ (cf.\ the proof of Proposition~\ref{prop-id-rec}), and for every $e \in \mathbb{N}$, let
\begin{align*}
{\Psi^{\ast}_e \left( x \right)}\ &:=\
{\left( \exists y \leqslant \rho_{\ast} \left( x \right) \right)}\,
{\left(
{R_{\ast} \left( x, y \right)} \wedge
{\forall u}\,
{\left(
{E_{\ast} \left( u, y \right)} \rightarrow
{\forall z}\,
{\left(
{\Theta \left( \underline{e}, {\dot{\eta} \left( u \right)}, z \right)}
\rightarrow
{\Pi^0_1  \text{-} \mathrm{SAT} \left( u, z \right)}
\right)}
\right)}
\right)} ,\\
{\Psi^{\bullet}_e \left( x \right)}\ &:=\
{\left( \exists y \leqslant \rho_{\bullet} \left( x \right) \right)}\,
{\left(
{R_{\bullet} \left( x, y \right)} \wedge
{\forall u}\,
{\left(
{E_{\bullet} \left( u, y \right)} \rightarrow
{\forall z}\,
{\left(
{\Theta \left( \underline{e}, {\dot{\eta} \left( u \right)}, z \right)}
\rightarrow
{\Pi^0_1  \text{-} \mathrm{SAT} \left( u, z \right)}
\right)}
\right)}
\right)} .
\end{align*}
Evidently, there exists a computable two-place function $f$ that, given any $e$ and $n$, returns the G\"{o}del number of a $\Pi^0_1$-formula that is
equivalent to
\[
{
{\dot{\eta} \left( x \right) = \underline{n}} \wedge
\left( \Psi^{\ast}_e \left( x \right) \vee \Psi^{\bullet}_e \left( x \right) \right)
} .
\]
Now let $h$ be a total computable function from $\mathbb{N}$ to $\mathbb{N}$ which satisfies
\[
{\ae_{h \left( e \right)} \left( n \right)}\ =\ {f \left( e, n \right)}
\]
for all $e \in \mathbb{N}$ and $n \in \mathcal{N}$. By the recursion theorem, we have $\ae_{h \left( c \right)} = \ae_c$ for some $c \in \mathbb{N}$. As can easily be checked, the function $\ae_c$ does the job.
\end{proof}

\begin{Cor}
Suppose $\mathcal{C} = \varnothing$. Then the derivability problem for $! \mathrm{ACT}_{\omega}$ is $\Pi^0_1$-bounded.
\end{Cor}

\begin{proof}
Take $g$ to be the function that exists by Proposition~\ref{Prop-Pi01-levels}. Then
\[
\left\{
{\sharp s} \mid s ~ \text{is derivable in} ~ ! \mathrm{ACT}_{\omega}
\right\}
\]
can be defined in $\mathfrak{N}$ by $\Pi^0_1 \text{-} \mathrm{SAT} \left( x, {g \left( \dot{\eta} \left( x \right) \right)} \right)$ (which is equivalent
to a $\Pi^0_1$-formula).
\end{proof}


\section{Conclusion and future work}\label{sec-conclusion}

We have established exact complexity bounds for infinitary action logic extended with subexponentials, at least one of which allows the non-local contraction rule (i.e.\ $\Cc \ne \varnothing$). The bounds are established both in the sense of complexity of the derivability problem and in the sense of the closure ordinal of the corresponding derivability operator.
In the case where no subexponential allows contraction (i.e.\ $\Cc = \varnothing$), we have also established a tight complexity bound for the derivability problem, namely, we have shown that it is $\Pi^0_1$-complete. As for the closure ordinal, finding its exact value in the case when $\Cc = \varnothing$ remains an open problem.  However, we have established an upper bound of $\omega^\omega$, which allowed us to prove the upper $\Pi^0_1$ complexity bound in a more invariant way than the ones used for infinitary action logic before.

In the view of the huge complexity gap between $\Pi^0_1$ and $\Pi_1^1$, it appears to be an interesting direction of research to find fragments of ${!}\mathrm{ACT}_\omega$ of intermediate complexity. While the present article was under review, this study began in~\cite{Kuzn2021Tableaux} by considering a system where $\Cc$ is non-empty, but $^*$ is not allowed under ${!}^c$ for $c \in \Cc$. This system is $\Pi^0_2$-hard and $\Delta_1^1$-bounded~\cite{Kuzn2021Tableaux}; thus it is indeed an intermediate one. Exact complexity of this fragment, in terms of both the derivability problem and the corresponding operator, however, remains an open question.

The situation for such weaker fragments of ${!}\mathrm{ACT}_\omega$ still requires further study.

%

Another closely related system arises if one replaces contraction with a weaker rule called {\em multiplexing:}
$$
\infer[(\mathrm{mult})]
{\Gamma, {!} A, \Delta \yields C}
{\Gamma, \overbrace{A, \ldots, A}^{\text{$n$ times}}, \Delta \yields C}
$$
An extension of the multiplicative-additive Lambek calculus with a subexponential allowing such a rule, and another subexponential allowing only permutation, was considered in~\cite{KanKuzSceIJCAR2020}. Even without the Kleene star, this system is $\Sigma_1^0$-hard~\cite{KanKuzSceIJCAR2020}. Therefore, in the presence of the Kleene star governed by an $\omega$-rule, the system will be strictly above $\Pi_1^0$. On the other hand, we conjecture that the closure ordinal for this system is less than or equal to $\omega^\omega$. This would give a hyperarithmetical upper bound on the complexity of the derivability problem; thus this will be another example of a system with intermediate complexity. Exact complexity, again, is an open problem.

\medskip
\paragraph{Acknowledgments.}
The authors are grateful to the referees for valuable comments and suggestions.
The work was supported by the Russian Science Foundation, in cooperation with the Austrian Science Fund, under grant RSF--FWF 20-41-05002.


\smallskip
\begin{flushleft}

\end{flushleft}


\vspace{10mm}
\noindent
{\small
\noindent
{\scshape Stepan L.\ Kuznetsov}\smallskip\\
Steklov Mathematical Institute\\
8 Gubkina St.,\\
Moscow 119991, Russia\smallskip\\
{\ttfamily sk@mi-ras.ru}\smallskip\\
\href{https://orcid.org/0000-0003-0025-0133}{\sffamily ORCID:\,0000-0003-0025-0133}\\

\bigskip
\noindent
{\scshape Stanislav O.\ Speranski}\smallskip\\
Steklov Mathematical Institute\\
8 Gubkina St.,\\
Moscow 119991, Russia\smallskip\\
{\ttfamily katze.tail@gmail.com}\smallskip\\
\href{https://orcid.org/0000-0001-6386-5632}{\sffamily ORCID:\,0000-0001-6386-5632}
}


\end{document}